\documentclass[10pt]{article} %sig-alternate} %{sig-alternate-10pt}

  \makeatletter
  \let\@copyrightspace\relax
  \makeatother

\usepackage{graphicx}   % need for figures
\usepackage{subfig}  % use for side-by-side figures
\usepackage{xspace}
\usepackage[usenames,dvipsnames]{color}
\usepackage{epsfig} 
\usepackage{algorithm}
\usepackage{algorithmic} 
\usepackage{hyperref}   % use for hypertext links, including those to external documents and URLs
\usepackage{enumerate}
\usepackage{fullpage}
\usepackage{amsmath, amssymb, amsthm}

\newcommand{\lin}{H}
\newcommand{\query}{h}
\newcommand{\newquery}{b}
\newcommand{\newlin}{B}
\newcommand{\gk}{\mathbf{\newquery_k}}
\newcommand{\XOR}{{XOR}}  
\newcommand{\EQUIV}{{EQUIV}}

\newcommand{\decDNNF}{{decision-DNNF}}  
\newcommand{\DecDNNF}{{Decision-DNNF}} 
\newcommand{\decDNNFn}{{DLDD}}  
\newcommand{\FBDD}{{FBDD}}  
\newcommand{\OBDD}{{OBDD}}  
\newcommand{\ANDFBDD}{{AND-FBDD}}

\newcommand{\F}{\mathcal{F}}
\newcommand{\D}{\mathcal{D}}
\newcommand{\OO}{\mathcal{O}}
\newcommand{\PP}{\mathcal{P}}
\newcommand{\noop}{\texttt{no-op}}
\newcommand{\opi}{unit}
\newcommand{\opis}{unit}
\newcommand{\units}{U}
\newcommand{\hkunits}{U_k}

\newcommand{\hk}{\mathbf{\query_k}}
\newcommand{\Hk}{\mathbf{\lin_k}}
\newcommand{\vecX}{\mathbf{X}}

\newcommand{\cut}[1]{}
\newcommand{\commentresolved}[1]{}

\newcommand{\ie}{{\em i.e.}\xspace}
\newcommand{\eg}{{\em e.g.}\xspace}

\newcommand{\set}[1]{\{#1\}}                    % Set (as in \set{1,2,3}).
\newcommand{\setof}[2]{\{{#1}\mid{#2}\}}        % Set (as in \setof{x}{x>0}).
\usepackage{aliascnt}  		% ``hyperref’s \autoref command does not work well with theorems that share a counter:
\newtheorem{theorem}{Theorem}[section]          	% Theorem environment.
\newaliascnt{lemma}{theorem}				% 1 alias counter
\newtheorem{lemma}[lemma]{Lemma}              	% Lemma environment.
\aliascntresetthe{lemma}  					% 3 set
\newaliascnt{conjecture}{theorem}			% 1 alias counter
\aliascntresetthe{conjecture}  				% 3 set
\newaliascnt{remark}{theorem}				% 1 alias counter
              
\aliascntresetthe{remark}  					% 3 set
\newaliascnt{corollary}{theorem}			% 1 alias counter
\newtheorem{corollary}[corollary]{Corollary}      % Corollary environment.
\aliascntresetthe{corollary}  				% 3 set
\newaliascnt{definition}{theorem}			% 1 alias counter
\newtheorem{definition}[definition]{Definition}    % Definition environment.
\aliascntresetthe{definition}  				% 3 set
\newaliascnt{proposition}{theorem}			% 1 alias counter
\newtheorem{proposition}[proposition]{Proposition}  % proposition environment.
\aliascntresetthe{proposition}  				% 3 set
\newaliascnt{example}{theorem}			% 1 alias counter
\newtheorem{example}[example]{Example}  	% 2 environment.
\aliascntresetthe{example}  				% 3 set
\usepackage{relsize}

\newcommand\ignore[1]{\relax}
\newcommand\myeq{\mathrel{\overset{\makebox[0pt]{\mbox{\normalfont\tiny\sffamily def}}}{=}}}

\begin{document}
\begin{sloppypar}

% ****************** TITLE ****************************************

%\title{Model Counting on Query Expressions is Strictly More Powerful
% than on Propositional Formulas}
\title{Model Counting of Query Expressions: \\Limitations of Propositional Methods}

\author{
Paul Beame~~~~~~~Jerry Li~~~~~~~Sudeepa Roy~~~~~~~Dan Suciu
\vspace{3mm}\\
\small{University of Washington, Seattle, WA}\\
\small{\emph{\{beame, jerryzli, sudeepa, suciu\}@cs.washington.edu}}
}

%\author{
%\alignauthor
%Jerry Li \and Sudeepa Roy \and Dan Suciu \\
%       \affaddr{University of Washington, Seattle, WA}\\
%       \email{\{sudeepa,jerryzli,suciu\}@cs.washington.edu}
%}

%\author{
% You can go ahead and credit any number of authors here,
% e.g. one 'row of three' or two rows (consisting of one row of three
% and a second row of one, two or three).
%
% The command \alignauthor (no curly braces needed) should
% precede each author name, affiliation/snail-mail address and
% e-mail address. Additionally, tag each line of
% affiliation/address with \affaddr, and tag the
% e-mail address with \email.
%
% 1st. author
%\and  % use '\and' if you need 'another row' of author names
%}
% There's nothing stopping you putting the seventh, eighth, etc.
% author on the opening page (as the 'third row') but we ask,
% for aesthetic reasons that you place these 'additional authors'
% in the \additional authors block, viz.
%\additionalauthors{Additional authors: John Smith (The Th{\o}rv\"{a}ld Group,
%email: {\texttt{jsmith@affiliation.org}}) and Julius P.~Kumquat
%(The Kumquat Consortium, email: {\small \texttt{jpkumquat@consortium.net}})}
\date{\today}
% Just remember to make sure that the TOTAL number of authors
% is the number that will appear on the first page PLUS the
% number that will appear in the \additionalauthors section.

\date{}
\maketitle

\begin{abstract}
Query evaluation in tuple-independent probabilistic databases 
is the problem of computing the probability of an answer to a query
given independent probabilities of the individual tuples in a database instance.
There are two main approaches to this problem: 
(1) in \emph{grounded inference} one first obtains the
lineage for the query and database instance as a Boolean formula,
then performs
weighted model counting on the lineage (i.e., computes the probability of the
lineage given probabilities of its independent Boolean variables);
(2) in methods known as \emph{lifted inference} or \emph{extensional
query evaluation}, one exploits the high-level structure of the query as a
first-order formula.
Although it is widely believed that lifted inference is strictly more powerful
than grounded inference on the lineage alone, no formal separation has
previously been shown for query evaluation. 
In this paper we show such a formal separation for the first time.
\par
We exhibit a class of queries for which model counting can be done in polynomial
time using extensional query evaluation, whereas the 
algorithms used in state-of-the-art exact model counters on their lineages
provably require exponential time.
Our lower bounds on the running times of these exact model counters follow
from new exponential size lower bounds on the kinds of $d$-DNNF
representations of the lineages that these model counters (either explicitly
or implicitly) produce.  Though some of these queries have been studied before,
no non-trivial lower bounds on the sizes of these representations 
for these queries were previously known.

\end{abstract}

\section{Introduction}
Model counting is the problem of computing the number, $\# \Phi$, of
satisfying assignments of a Boolean formula $\Phi$. 
In this paper we are concerned with the weighted version of model counting, 
which is the same as the probability computation problem on independent 
random variables. 
Although model counting is \#P-hard in general 
(even for formulas where satisfiability is easy to check) \cite{DBLP:journals/siamcomp/Valiant79}, there have been major
advances in practical algorithms that compute exact, weighted model
counts for many relatively complex formulas.  Exact model counting for
propositional formulas (see~\cite{DBLP:series/faia/GomesSS09} for a
survey) are based on extensions of backtracking search using the {\em
  DPLL} family of algorithms \cite{DavisPutnam60, DPLL62} that were
originally designed for satisfiability search.

The study of (weighted) model counting in this paper is %largely 
motivated by
query evaluation in tuple-independent probabilistic databases (see,
\eg, \cite{DBLP:series/synthesis/2011Suciu}):
given a (fixed) Boolean query $Q$, and 
independent probabilities of the individual tuples in a probabilistic
database $D$, compute $\Pr[Q(D)]$, the 
probability that $Q$ is true on a random instance of $D$. 
Here the propositional formula $\Phi$ is the grounding of the first-order
formula representing the query $Q$ on 
database $D$, and is called the \emph{lineage}.
We call this the \emph{grounded inference} approach:
first compute $\Phi$, then perform model counting on $\Phi$ at the
propositional level  to compute $\Pr[Q(D)]$.
			
The mismatch between the high
level representation as a first-order formula and the low level
of propositional inference was noted early on, and has given rise to
various techniques that operate at the first-order level, which are
collectively called {\em lifted inference} in statistical relational
models (see~\cite{jaeger-broeck-2012} for an extensive discussion), or
{\em extensional query evaluation} in probabilistic
databases~\cite{DBLP:series/synthesis/2011Suciu}.  These methods
exploit the high level structure of the first-order formula for the query.

It is widely believed that lifted inference, or extensional
query evaluation, is strictly more powerful than grounded inference that does
not take advantage of the high level structure.   While there have been 
examples in other contexts where provable separations have been
shown~(e.g.,~\cite{sabharwal:symchaff-journal}), %  \red{A}),
no formal separation has previously been shown in the context of query
evaluation.  We show such a formal separation for the first time.

We describe a class of queries and prove a formal statement which implies that
grounded inference in the form of
current propositional model counting algorithms requires exponential time
on any member in the class.
On the other hand, model counting for each query in this class can be done in
polynomial time in the size of the domain using the inclusion/exclusion
formula applied to the query
expression~\cite{DBLP:conf/pods/DalviSS10,DBLP:series/synthesis/2011Suciu},
which is a form of extensional query evaluation, or lifted inference.

We first review the state of the art for propositional model counting
algorithms that are based on extensions of DPLL family of algorithms
\cite{DavisPutnam60, DPLL62}.
Extensions include caching the results of
solved sub-problems \cite{MajercikLittman98}, dynamically decomposing
residual formulas into components (Relsat \cite{Bayardo+00-modelcounter}) and
caching their counts (\cite{bacchusdalmaopitassi}),
and applying dynamic component caching
together with conflict-directed clause learning (CDCL) to further
prune the search (Cachet \cite{cachet} and sharpSAT \cite{sharpSAT}).
In addition to DPLL-style algorithms that compute the counts on the fly,
model counting has been addressed through a complementary approach, known as
{\em knowledge compilation}, which converts the input formula into a
representation of the Boolean function that the formula defines and from
which the model count can be computed efficiently in the size of the
representation~\cite{Darwiche01JACM, Darwiche01, HuangDJAIR07,Dsharp}. 
Efficiency for knowledge compilation depends both on the
size of the representation and the time required to construct it.
These two approaches are quite related. 
As noted in~{\tt c2d}~\cite{HuangDJAIR07} (based on component caching) and
Dsharp~\cite{Dsharp} (based on sharpSAT), the traces of all the
DPLL-style methods yield knowledge compilation algorithms that can
produce what are known as {\em \decDNNF}
representations~\cite{HuangD-IJCAI05, HuangDJAIR07}, a syntactic
subclass of $d$-DNNF representations~\cite{Darwiche01, DarwicheM2002};

Indeed, all the methods for exact model counting surveyed
in~\cite{DBLP:series/faia/GomesSS09} (and all others of which we are
aware) can be converted to knowledge compilation algorithms that
produce \decDNNF\ representations, without any significant increase in
their running time.
A \decDNNF\ is a rooted DAG where each node
either tests a variable $Z$ and has two outgoing edges corresponding
to $Z=0$ or $Z=1$, or is an AND-node with two sub-DAGs that do not
test any variable in common;
these naturally
correspond to the two types of operations in any modern DPLL-style
algorithm: Shannon expansion on a variable $Z$, or partitioning the
formula into two disconnected components.  
We will refer to the DPLL-style algorithms and knowledge compilation methods
used in the state-of-the-art model
counters as \emph{\decDNNF-based model counting algorithms}. 

Therefore, we obtain our lower bounds on propositional model counting
algorithms by showing that, for every query in the
class we define, every \decDNNF\ for the lineage of that query is
exponentially large in the size of the domain; this implies that all
current model counting algorithms will run in exponential time
on the propositional grounding of that query.

\vspace{1ex}
\noindent\textbf{Our Contributions.~~} 
We first consider a well-studied family of simple queries whose
associated (weighted) model counting problems are known to
be \#P-hard~\cite{DBLP:journals/jacm/DalviS12}.
These queries have lineages that are simple 2-DNF formulas.
We %begin by proving 
prove exponential lower bounds of the form
$2^{\Omega(\sqrt{n})}$ on the sizes of \decDNNF\ 
representations of %any of these 
these lineages, % for any of these queries, 
which is the first non-trivial
\decDNNF\ lower bounds for these queries (\autoref{TH:1}).

We obtain these lower bounds by first proving that any {\em FBDD}
representation for any of these formulas requires at least $2^{n-1}/n$ size for 
a domain of size $n$; then using our recent result~\cite{Beame+13-uai} showing
that every \decDNNF\ of size $N$ can be
converted into an \FBDD\ of size at most $N 2^{\log^2 N}$. (FBDDs, 
{\em free binary decision diagrams}, also known as {\em read-once branching
programs}, are a restricted subclass of \decDNNF s
without any AND-nodes).
A much weaker lower bound of $2^{\Omega(\log^2 n)}$ was shown
previously~\cite{DBLP:journals/mst/JhaS13} for the \FBDD\ size of one of these
queries, but this is insufficient to yield any \decDNNF\ lower bound using
our %above 
translation.

We also strengthen the conversion from~\cite{Beame+13-uai} to one with the
same complexity that applies to a new, more
general class of representations than \decDNNF s  (\autoref{TH:2}).
We call this class {\em decomposable
logic decision diagrams} (\decDNNFn s). 
\decDNNFn s can better represent ways of using \decDNNF-based algorithms for
query lineages, which naturally are given as DNF formulas rather than the CNF
formulas that are the natural input for DPLL algorithms.
%This result enables us to derive lower bounds on the running time of 
%\decDNNF-based algorithms for query lineages for which can be expressed as a $k$-DNF formula
%with a fixed $k$ for a fixed query, whereas the
%DPLL-style algorithms were originally 
%proposed for CNF formulas. 

%As should not be surprising given that model counting for the
%queries considered in \autoref{TH:1} is \#P-hard, we do not have corresponding efficient algorithms for 
%model counting for these queries, even using lifted methods.
Given that model counting for the
queries considered in \autoref{TH:1} is \#P-hard even using lifted methods, it should not be surprising that
we do not have efficient algorithms for 
model counting for these queries, .
So \autoref{TH:1} does not lead to a separation between lifted and grounded
inference.
The separation we derive follows by substantially generalizing the class of
queries for which the lower bounds on the sizes of \FBDD s and \decDNNFn s hold.
Each of the queries in \autoref{TH:1} is a disjunction of a number of elementary
queries. 
We strengthen our lower bounds by showing that, if we generalize
the disjunction operation to {\em any}
Boolean combination of the same elementary queries 
such that the query $Q$ built this way depends on all of the elementary queries,
then the exponential lower bounds also hold for $Q$ (\autoref{TH:3}).
To prove this surprising result we show that every \FBDD\ for such
a Boolean combination can, with a small increase in size, be converted into
an \FBDD\ that simultaneously represents the lineages of all of its
constituent elementary queries, and therefore into an \FBDD\ for the
disjunction query for which the previous lower bounds hold.
%We note that a result converting an \FBDD\ for one particular query
%$Q_W$ into an \FBDD\ for one of these disjunctions was shown
%in~\cite{DBLP:journals/mst/JhaS13}, but this did not yield a general technique.

To obtain the separation between the \decDNNF-based algorithms for
model counting and lifted inference on a query $Q$, 
we apply a
characterization given in~\cite{DBLP:journals/jacm/DalviS12} for Union of
Conjunctive Queries (UCQ), which gives a specific property 
(entirely) based on the structure of $Q$ that allows exact
model counting in time polynomial in the size of the database. 
This yields a $2^{\Omega(\sqrt{n})}$ versus $n^{O(1)}$
separation between these propositional and lifted methods for weighted
model counting for a wide variety of such queries $Q$ (\autoref{TH:MAIN-SEP}).

\vspace{1ex}
\noindent\textbf{Roadmap.~~} We discuss some useful knowledge compilation
representations in \autoref{sec:fbdds}. 
In \autoref{sec:main-results}, we describe our main results
which are proved in the following Sections \ref{sec:h0h1} and \ref{sec:hk}.
We discuss related issues in \autoref{sec:conclusions}.

\section{Background}\label{sec:fbdds}
In this section we review the knowledge compilation representations used 
in the rest of the paper.\\

\textbf{\FBDD s.~~} An \FBDD\ \footnote{\FBDD s are also known as \emph{Read Once Branching Program}s.}
is a rooted directed acyclic graph (DAG) $\F$
that computes $m$ Boolean functions $\mathbf{\Phi} = (\Phi_1,\ldots, \Phi_m)$. $\F$ has two
kinds of nodes: {\em decision nodes}, which are labeled by a Boolean
variable $X$ and have two outgoing edges labeled 0 and 1; and
\emph{sink nodes} labeled with an element from $\set{0, 1}^m$.  
Every path from the root to some sink node may test a Boolean variable $X$
at most once.  For each assignment $\theta$ on all the Boolean variables,
$\mathbf{\Phi}[\theta] = (\Phi_1[\theta], \ldots, \Phi_m[\theta]) = L$, where
$L$ is the label of the unique sink node reachable by following the
path defined by $\theta$.  The \emph{size} of the \FBDD\ $\F$ is the number 
nodes in $\F$.  Typically $m = 1$,  but we will also consider \FBDD s $\F$ with $m>1$ and call $\F$ a {\em multi-output
  \FBDD}.
\par
For every node $u$, the sub-DAG of $\F$ rooted at $u$, denoted $\F_u$,
computes $m$ Boolean functions $\mathbf{\Phi}_u$ defined as
follows.  If $u$ is a decision node labeled with $X$ and has children
$u_0,u_1$ for 0- and 1-edge respectively, then $\mathbf{\Phi}_u = (\neg X) \mathbf{\Phi}_{u_0} \vee X
\mathbf{\Phi}_{u_1}$; if $u$ is a sink node labeled $L \in
\set{0,1}^m$, then $\mathbf{\Phi}_u = L$.  $\F$ computes
$\mathbf{\Phi} = \mathbf{\Phi}_r$ where $r$ is the root.  The probability of each of
the $m$ functions can be computed in time linear in the size of the
\FBDD\ using a simple dynamic program: $\Pr[\mathbf{\Phi}_u] = (1 -
p(X)) \Pr[\mathbf{\Phi}_{u_0}] + p(X) \Pr[\mathbf{\Phi}_{u_1}]$. 
\par
For our purposes, it will also be useful to consider FBDDs with {\em
  no-op nodes}.  A no-op node is not labeled by any variable, and has
a single child; the meaning is that we do not test any variable, but
simply continue to its unique child.  Every FBDD with no-op nodes can
be transformed into an equivalent FBDD without no-op nodes, by simply
skipping over the no-op node.\\

\textbf{\DecDNNF s.~~}
A \decDNNF \footnote{A \decDNNF is a special case of both an \ANDFBDD\ (the AND nodes have no restriction)
\cite{wegener2000book} and a d-DNNF \cite{Darwiche01}, which are
different kinds of circuits used in knowledge compilation.
see~\cite{Beame+13-uai}  for a discussion.}
 $D$ generalizes an \FBDD\
allowing \emph{decomposable AND-nodes} in addition to decision-nodes,
\ie, any AND-node $u$ must satisfy the restriction that, for its two
children $u_1, u_2$, the sub-DAGS $\D_{u_1}$ and $\D_{u_2}$ do not
mention any common Boolean variables. The function $\mathbf{\Phi_u}$ is defined as $\mathbf{\Phi_u} =
\mathbf{\Phi_{u_1}} \wedge \mathbf{\Phi_{u_2}}$, and its probability 
is computed as $\Pr[\mathbf{\Phi_{u}}] = \Pr[\mathbf{\Phi_{u_1}}] \cdot \Pr[\mathbf{\Phi_{u_2}}]$.
In a \decDNNF, similar to \FBDD s,
any Boolean variable can be tested at most once along any path from the root to any sink.\\

\textbf{\decDNNFn s.~~}
In this paper we introduce \emph{Decomposable Decision Logic Diagrams} or \decDNNFn s 
by further generalizing a \decDNNF. 
	%A \decDNNFn\ can have decision nodes (tested at most once along any path) and 0- or 1- sinks like \decDNNF s.
	%But it can have
A \decDNNFn\ can also have
NOT-nodes $u$ having a unique child $u_1$, 
%and 
%instead of only having decomposable AND-nodes, it can have 
and decomposable %AND-, 
OR-, \XOR-, and \EQUIV-nodes similar to decomposable AND-nodes\footnote{These 
four nodes along with NOT-nodes can capture all possible non-constant 
functions on two Boolean variables}:
%$u$ with two
%children $u_1, u_2$ such that the sub-DAGS $\D_{u_1}$ and $\D_{u_2}$ do not
%mention any common Boolean variables. 
(i) for a NOT-node, 
$\mathbf{\Phi_u} = \neg \mathbf{\Phi_{u_1}}$, and 
$\Pr[\mathbf{\Phi_{u}}] = 1 - \Pr[\mathbf{\Phi_{u_1}}]$;
(ii) for an OR-node, $\mathbf{\Phi_u} = \mathbf{\Phi_{u_1}} \vee \mathbf{\Phi_{u_2}}$, and 
$\Pr[\mathbf{\Phi_{u}}] = 1 - (1 - \Pr[\mathbf{\Phi_{u_1}}]) \cdot (1 - \Pr[\mathbf{\Phi_{u_2}}])$;
(iii) for an \XOR-node, $\mathbf{\Phi_u} = \mathbf{\Phi_{u_1}} \cdot \neg \mathbf{\Phi_{u_2}}$ $\vee$ $\neg \mathbf{\Phi_{u_1}} \cdot \mathbf{\Phi_{u_2}}$, and
(iv) for an \EQUIV-node, $\mathbf{\Phi_u} = \mathbf{\Phi_{u_1}} \cdot \mathbf{\Phi_{u_2}}$ $\vee$ $\neg \mathbf{\Phi_{u_1}} \cdot \neg \mathbf{\Phi_{u_2}}$
(again $\Pr[\mathbf{\Phi_{u}}]$ can easily be computed from $\Pr[\mathbf{\Phi_{u_1}}], \Pr[\mathbf{\Phi_{u_2}}]$).
%$\Pr[\mathbf{\Phi_{u}}] = \Pr[\mathbf{\Phi_{u_1}}] \cdot (1 - \Pr[\mathbf{\Phi_{u_2}}]) + (1 - \Pr[\mathbf{\Phi_{u_1}}])\Pr[\mathbf{\Phi_{u_2}}])$;
%$\mathbf{\Phi_u} = \mathbf{\Phi_{u_1}} \cdot \mathbf{\Phi_{u_2}}$ $\vee$ $\neg \mathbf{\Phi_{u_1}} \cdot \neg \mathbf{\Phi_{u_2}}$, and
%$\Pr[\mathbf{\Phi_{u}}] = \Pr[\mathbf{\Phi_{u_1}}] + \Pr[\mathbf{\Phi_{u_2}}])$.
Hence the probability of the formula can still be computed in time linear in $\D$.\\

\textbf{Conversion of a \decDNNFn\ into an equivalent \FBDD.~~}
The trace of
any DPLL-based algorithm with caching and components is a \decDNNF. Therefore
any lower bound on the size of \decDNNF s represents a lower
bound on the running time of modern model counting algorithms.
We have proven recently the first lower bounds on \decDNNF s~\cite{Beame+13-uai}.  
However, model counting algorithms were designed
for CNF expressions: for example, the component analysis partitions
the clauses into two disconnected components (without common
variables), then computes the probability as $\Pr[\Phi_1 \wedge
\Phi_2] = \Pr[\Phi_1] \Pr[\Phi_2]$.  In order to run such an algorithm
on a DNF expression %(like $\Theta_0$ or $\Theta_1$) 
(which are more related to lineages in databases) one would
naturally first apply a negation, %, and obtain a CNF expression.  
which
transforms the formulas into CNF.  This suggest a simple extension of
such algorithms: allow the application of the negation operator at any
step.  The trace now also has NOT-nodes and therefore is a special case of \decDNNFn s. 
But we prove our first result in the paper for general \decDNNFn s:

\begin{theorem} \label{TH:2} For any \decDNNFn\ $\D$ with $N$ nodes
  there exists an equivalent \FBDD\ $\F$ computing the same formula as
  $\D$, with at most $N 2^{\log^2 N}$ nodes (at most quasi-polynomial increase in size).
\end{theorem}

In~\cite{Beame+13-uai} we have proven a similar result with the same bound
for \decDNNF s; now we strengthen it to \decDNNFn s.
The proof of Theorem~\ref{TH:2} appears in Section~\ref{sec:proof-thm-decdnnfn}.

%using a relatively minor extension to 
%extending the proof in \cite{Beame+13-uai}; 
%the extension is rather simple and appears in Appendix~\ref{app:proof-thm-decdnnfn}.  

				%In the next section we discuss how this quasi-polynomial conversion
			%combined with the exponential lower bound results on the size of \FBDD s shown in this paper leads to exponential 
			%lower bounds on the size of \decDNNFn s, and therefore on the running time of
			%the state of the art model counting algorithms.
			
\section{Main Results}
\label{sec:main-results}

Here we formally state our main results and discuss their implications,
and defer the proofs to the following sections. But first we introduce some %simple
elementary
queries that work as building blocks for the class of queries considered in these results \cite{DBLP:journals/jacm/DalviS12, DBLP:journals/mst/JhaS13}:
\par

Let $[n]$ denote the set $\{1,\ldots, n\}$. 
Fix $k > 0$ and consider the following set of $k+1$ Boolean queries
$\hk = (\query_{k0}, \cdots, \query_{kk})$, where 
\begin{align*}
  \query_{k0} & =  \exists x_0 \exists y_0\ R(x_0) \land S_1(x_0,y_0) \\
	\query_{k\ell} & =  \exists x_{\ell} \exists y_{\ell}\ S_{\ell}(x_{\ell},y_{\ell})\land S_{{\ell}+1}(x_{\ell},y_{\ell})\qquad \forall \ell \in [k-1] \\
  %\query_{k1} =  \exists x_1 \exists y_1 S_1(x_1,y_1),S_2(x_1,y_1) \\
  %\query_{k2} =  \exists x_2 \exists y_2 S_2(x_2,y_2),S_3(x_2,y_2) \\
          %& \cdots \\
  \query_{kk} & =  \exists x_k \exists y_k\ S_k(x_k,y_k)\land T(y_k)
\end{align*}

Fix a domain size $n > 0$; for each $i, j \in [n]$,
let $R(i)$, $S_1(i,j), \ldots, S_k(i,j)$, $T(j)$ be Boolean variables 
representing potential tuples in the database. 
Then the corresponding {\em lineages}, the associated Boolean expressions for these queries 
	%the associated Boolean functions for these
	%queries over $[n]$, are given by
are \footnote{For simplicity, conjunctions 
in Boolean formulas are represented as products.}:
\begin{align*}
  \lin_{k0} = & \bigvee_{i,j \in [n]} R(i) S_1(i,j), \qquad \lin_{kk} = \bigvee_{i,j \in [n]} S_k(i,j) T(j),\\
  \lin_{k\ell} = & \bigvee_{i,j \in [n]} S_{\ell}(i,j) S_{\ell+1}(i,j)\qquad \forall \ell \in [k-1]    
\end{align*}
We %also 
define $\Hk=(\lin_{k0},\ldots,\lin_{kk})$. Two well-studied queries \cite{DBLP:journals/jacm/DalviS12}
that we will consider
in this section are given below:

\textbf{Query $\query_k$:~~} $\query_k$ is a disjunction on the queries in $\hk$:~
$\query_k = \query_{k0} \vee \query_{k1} \vee  \cdots \vee \query_{kk}$.
The lineage $\lin_k$ of $\query_k$ is given by
$\lin_k = \lin_{k0} \vee \lin_{k1} \vee \cdots \vee \lin_{kk}$.\\

\textbf{Query $\query_0$:~~} Also we define $\query_0$ that uses a single relation symbol $S$ in addition to $R$ and $T$:~
$\query_{0} = \exists x \exists y\ R(x)\land S(x,y)\land T(y)$.
$S$ is defined on Boolean variables $S(i, j)$, $i, j \in [n]$, and therefore the lineage $\lin_0$ of $\query_0$ is
$\lin_0 = \bigvee_{i, j \in [n]} R(i) S(i, j) T(j)$.

\subsection*{Lower bounds on \FBDD s for queries $\query_0$, $\query_k$}
Jha and Suciu~\cite{DBLP:journals/mst/JhaS13} previously showed
that every \FBDD\ for the lineage $\lin_1$ of $\query_1$ has size
$2^{\Omega(\log^2 n)}$.
Our first result improves this to an exponential lower bound, not just for
$\lin_1$ but also for $\lin_0$ and all $\lin_k$ for $k>1$:

\begin{theorem} \label{TH:1}
For every $n > 0$, any \FBDD\ for $\lin_0$ or $\lin_k$ for $k\ge 1$ has
$\geq 2^{(n-1)}/n$ nodes.
\end{theorem}

It is known that weighted model counting for both
$\lin_0$ and $\lin_k$ is \#P-hard~\cite{DBLP:journals/jacm/DalviS12}.
However, the lower bounds we show on these \FBDD\ sizes are absolute
(independent of any complexity theoretic assumption) and do not rely
on the \#P-hardness of the associated weighted model counting problems.
We give the proof in \autoref{sec:h0h1}.  
This improved bound is critical
for proving the %main result 
overall lower bound result in this paper (\autoref{thm:main}).  

While we do not need
$\query_0$ and $\lin_0$ in the rest of the paper, we include %its proof here
it in \autoref{TH:1} 
because
it is obtained in a fashion similar to that for $\lin_k$ and substantially
improves on a $2^{\Omega(\sqrt{n})}$ lower bound for $\lin_0$ from our previous work \cite{Beame+13-uai} which
%showed a  \FBDD\ size lower bound using a % related 
%known lower bound
was based on a result by of Bollig and Wegener~\cite{bollig98}\footnote{Bollig and Wegener defined a set
   $E_n \subseteq [n] \times [n]$ for which
   any \FBDD\ for the formula $\bigvee_{(i,j) \in E_n} R(i) T(j)$ requires
   size $2^{\Omega(\sqrt{n})}$ which obviously implies the same lower bound
   for $\lin_0$.  The set $E_n$ is given as follows:
   Assume that $n=p^2$ where $p$ is a prime number. 
   Each number $0 \leq i < n$ can be uniquely written as $i=a + bp$ where
   $0 \leq a,b < p$.  Then:
   $E_n = \setof{(i+1,j+1)}{i=a+bp, j=c+dp, c \equiv (a + bd) \mod p}$.
   }.
Our new lower bound improves this to the nearly optimal $2^{n-1}/n$.

We also note that our stronger lower bounds for $\lin_1$ give
instances of bipartite 2-DNF formulas that are simpler to describe than
those of~\cite{bollig98} but yield as good a lower bound on \FBDD\ sizes in
terms of their number of variables and even better bounds as a function
of their number of terms\footnote{In the formulas of~\cite{bollig98}, $p$
is analogous to $n$ in our formulas and theirs have $p^3$ terms, versus only
$2n^2$ for our formulas.}.

\subsection*{Lower bounds for \FBDD s for queries over  $\hk$}

\autoref{TH:1} gives a lower bound on $\query_k$, which is simply the logical
OR of the queries in $\hk$.
\autoref{TH:3} below generalizes this result by allowing queries that are
{\em arbitrary functions} of queries in $\hk$.

Let $f(\vecX)=f(X_0, X_1, \cdots, X_k)$ be an arbitrary Boolean function on
$k+1$
Boolean variables $\vecX = (X_0, \cdots, X_k)$, and $Q$ the
Boolean query $Q = f(\query_{k0}, \query_{k1}, \cdots, \query_{kk})$.  
%Define the
Clearly, the lineage of $Q$ 
is $f(\Hk)=f(\lin_{k0}, \lin_{k1}, \cdots, \lin_{kk})$.  
\begin{example}
If $f(X_0, X_1, \cdots, X_k) = \bigvee_{\ell = 0}^k X_{\ell}$, we get query 
$\query_k =\bigvee_{\ell = 0}^k h_{k\ell}$; its lineage is $\lin_k = \bigvee_{\ell = 0}^k \lin_{k\ell}$.
%
%For a simple example, denoting the query $H_k = \query_{k0}
%\vee \cdots \vee \query_{kk}$, then the lineage of $H_1$ is $\lin_1$ in
%\autoref{TH:1}.
%\footnote{The formula $\lin_0$ is the lineage of the
  %query $H_0 = \exists x \exists y R(x),S(x,y),T(y)$, which is not a
  %Boolean combination of $\HH_k$.}
\end{example}
The function $f$ \emph{depends on} a variable $X_\ell$,
$\ell \in \{0,\ldots,k\}$, if
there is an assignment $\mu_{\ell}$ on
the rest of the variables $\vecX \setminus \{X_{\ell}\}$ such that $f[\mu_{\ell}] = X_{\ell}$ or $\neg X_{\ell}$.

\begin{theorem} \label{TH:3} 
If $f$ depends on all $k+1$
variables $X_0, \cdots, X_k$, then any \FBDD\ $\F$ with $N$ nodes for the
  lineage of $Q = f(\query_{k0}, \cdots, \query_{kk})$ can be converted into a
  multi-output \FBDD\ for $(\lin_{k0}, \cdots, \lin_{kk})$ with $O(k 2^k n^3 N)$
  nodes.
  In particular, for $k\le \alpha n$ for any constant $\alpha<1$, 
  $\F$ requires at least $2^{\Omega(n)}$ nodes. 
\end{theorem}
\noindent
We prove the theorem in \autoref{sec:hk}.  The condition that $f$
depends on all variables is necessary
(see Sections \ref{sec:h0h1} and \ref{sec:hk}):
if $Q$ does not
depend on any one of the queries in $\hk$, then its lineage has an \FBDD\ of 
size linear in the number of Boolean variables.
\par
\autoref{TH:3} extends prior work in several ways. 
First, it is the first result showing exponential lower bounds on \FBDD s for a large class of queries.
Prior to \autoref{TH:3} the only known lower bound was the quasipolynomial lower bound for $\query_1$ \cite{DBLP:journals/mst/JhaS13}.
Second, although  a conversion of an \FBDD\ for a specific query $Q_W$
(described later in this section)
 into one for $\query_1$ was given
in~\cite{DBLP:journals/mst/JhaS13}, % to show lower bounds on the size of \FBDD s for $Q_W$, 
this %translation 
conversion did not extend to other queries. 
While we
were inspired by %the existence of 
that proof, the techniques we use in
\autoref{TH:3} are considerably more powerful, and uses new ideas 
which can be of independent interest to show lower bounds on the size 
of \FBDD s in general.
\par
We also extend the lower bound in \autoref{TH:3} by proving a dichotomy theorem for 
a slightly more general class of queries: any query in this class either has
a polynomial-time model counting algorithm, or all existing \decDNNF-based model counting
algorithms require exponential time. The details appear in Section~\ref{sec:mini:dichotomy}.

\subsection*{Lower Bounds for Model Counting Algorithms for Queries over $\hk$}

Theorems~\ref{TH:2}, \ref{TH:1}, and \ref{TH:3} together prove the 
%main result of our paper:
following lower bound result:

\begin{theorem}%[Main Result] 
\label{thm:main}%\label{thm:main} 
  If $Q$ is a Boolean
  combination of the queries in $\hk$ that depends on all $k+1$
  queries in $\hk$, then any \decDNNFn\ (and therefore any \decDNNF)
	for the lineage $\Theta$ of $Q$ has size
  $2^{\Omega(\sqrt n)}$. 
% Therefore, any \decDNNF-based algorithm will have 
% running time $2^{\Omega(\sqrt n)}$ to perform weighted model counting for $Q$.
\end{theorem}

\begin{proof}
  Let $N$ be the size of a \decDNNFn\ for $Q$.  By \autoref{TH:2}, $Q$
  has an \FBDD\ of size $N 2^{\log^2 N}$.  By \autoref{TH:3}, $H_1$
  has an \FBDD\ for size $2^{O(\log^2 N)}$, which has to be 
  $2^{\Omega(n)}$ by \autoref{TH:1}, implying that $N$ is
  $2^{\Omega(\sqrt{n})}$.
\end{proof}

Notice that in order to prove %the main result, 
\autoref{thm:main}, we needed the
strong exponential lower bound on the size of an \FBDD\ for $\lin_k$ in
\autoref{TH:1}: the prior quasipolynomial lower bound 
was not sufficient because of the quasipolynomial increase in size in
\autoref{TH:2} moving from \decDNNFn s to \FBDD s.

Since, as we discussed previously, current propositional exact weighted model
counting algorithms (extended with negation to handle DNFs)
without loss of generality 
yield \decDNNFn s of size at
most their running time, we immediately obtain:
\begin{corollary}\label{cor:main}
All current propositional exact model counting algorithms require
running time $2^{\Omega(\sqrt n)}$ to perform weighted model counting for any
query $Q$ that is a Boolean
  combination of the queries in $\hk$ and depends on all $k+1$
queries in $\hk$.
\end{corollary}

\subsection*{Propositional versus Lifted Model Counting}
%
		%The significance of \autoref{thm:main} and \autoref{cor:main}
		%is that they prove an exponential
		%lower bound on \decDNNF-based weighted model counting algorithms (\ie, all state of the art algorithms used in practical propositional model counters) 
		%on the lineage of any query $Q$ defined over $\hk$ that depends on all
		%coordinates in $\hk$.
		%In the case of one such query $h_k$, 
\autoref{thm:main} when applied to query $h_k$, $k \geq 1$, is not surprising:
\#P-hardness of $h_k$ makes it 
unlikely to have an efficient model counting algorithm.
However, there are many other query combinations %defined over all the functions in 
over $\hk$ for which lifted methods %that take advantage 
taking advantage of the high-level
structure yield polynomial-time %weighted 
model
counting and therefore outperform current propositional techniques.

%Here we show that if we even when model counting can be done in polynomial time on the
%query expression $Q$ (as we explain below).  
%This suggests that model counting
%on the query expression is likely to be more powerful than on its propositional
%grounding (although this theorem does not imply inexistence of such a technique).  

%There are simple examples where model counting for $Q$ can be done in
%time polynomial in the number of Boolean variables:
%\begin{example}
 %If $f(X_0, X_1, X_2) = X_0 X_1 \vee (\neg X_0) X_2$ one can compute the probability as $\Pr(Q) = \Pr(\query_{20} \wedge \query_{21}) +
%\Pr(\neg \query_{20}) \Pr(\query_{22})$, it can be verified that both expressions can be computed in polynomial time.  
%\end{example}

Consider the case when $Q=f(\hk)$ and $f$ is a monotone Boolean formula $f(X_0, \cdots, X_k)$,
and thus $Q$ is a UCQ query.  Here the cases when weighted model
counting for $Q$ can be done in polynomial time are entirely
determined by the structure of the query expression\footnote{The
  propositional formula $f$ describes the query expression $Q$, and
  should not be confused with the propositional grounding of $Q$ on the instance $R(i), S_{1}(i, j),$ $\cdots, S_{k}(i, j), T(j)$; $\ell \in [1, k-1]$, $i, j \in [n]$.}
$f$, and we review it here briefly
following~\cite{DBLP:series/synthesis/2011Suciu}.

To check if weighted model counting for %$Q = f(\query_{k0}, \cdots, \query_{kk})$
$Q$ is computable in polynomial time,
write $f$ as a CNF formula, $f = \bigwedge_i C_i$, where each (positive)
clause $C_i$ is
a set of propositional variables $C_i \subseteq \set{X_0, \cdots, X_k}$.
Define the lattice $(L,\leq)$, where $L$ contains all subsets
$u \subseteq \mathbf{X}$ that
are a union of clauses $C_i$, and the order relation is given by $u
\leq v$ if $u \supseteq v$.  The maximal element of the lattice is
$\emptyset$, (we denote it $\hat 1$), while the minimal element is
$\mathbf{X}$ (we denote it $\hat 0$).  The M\"obius function on the
lattice $L$, $\mu : L \times L \rightarrow \mathbf{R}$, is defined as
$\mu(u,u)=1$ and $\mu(u,v)$ = $-\sum_{u<w \leq v} \mu(w,v)$
\cite{stanley-combinatorics-1997}.  The following
holds~\cite{DBLP:series/synthesis/2011Suciu}: if $\mu(\hat 0, \hat 1)=
0$, then weighted model counting for $Q$ can be done in time polynomial
in $n$ (using in the inclusion/exclusion formula on the CNF); if
$\mu(\hat 0, \hat 1) \neq 0$, then the weighted model counting problem for
$Q$ is \#P-hard. 

\begin{figure}[t]
  \centering
%\hfill
  %\includegraphics[scale=0.2]{figs/qW}
	\includegraphics[scale=0.2]{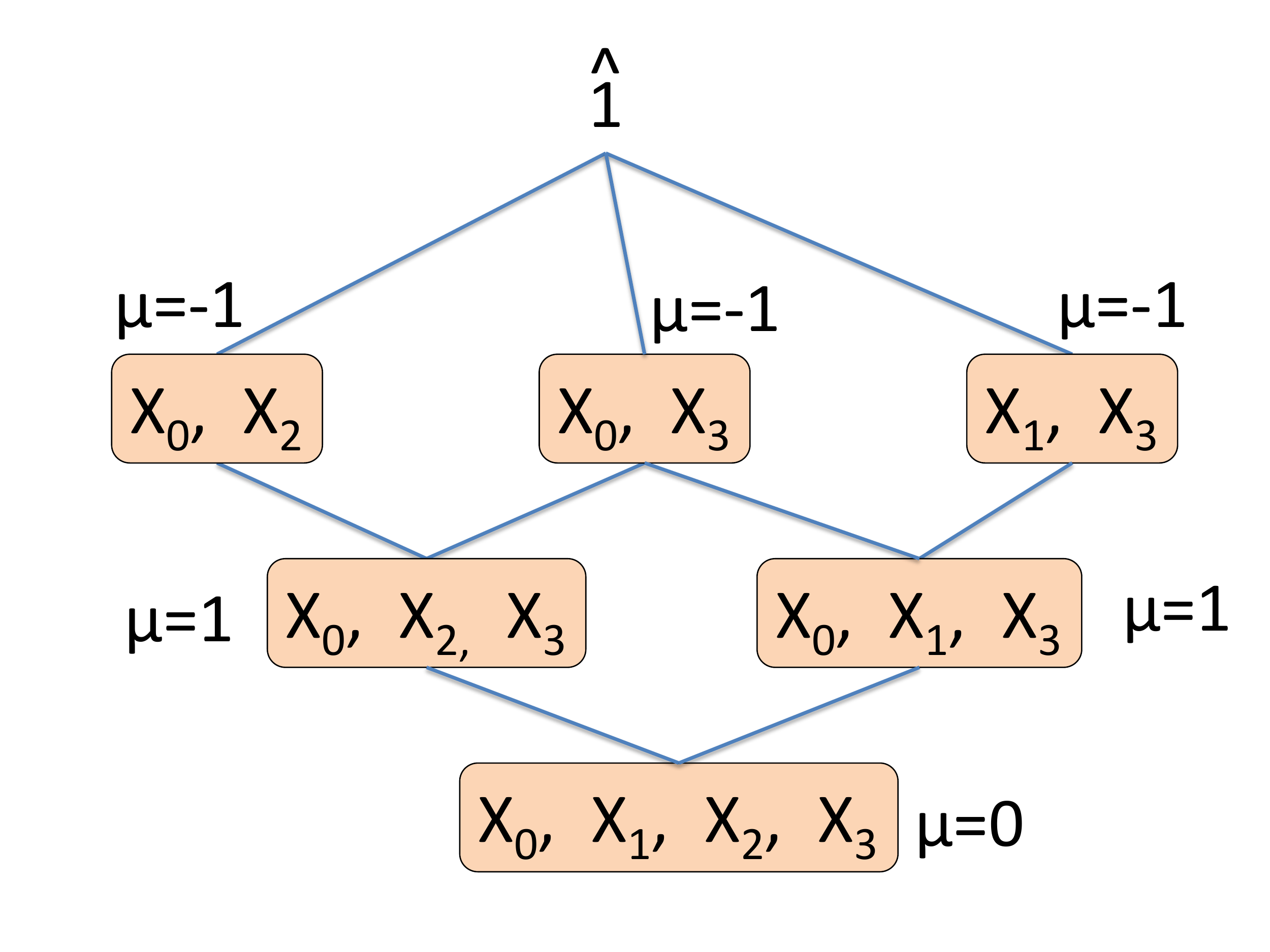}
~~~~~~~~~~~~~
	\includegraphics[scale=0.2]{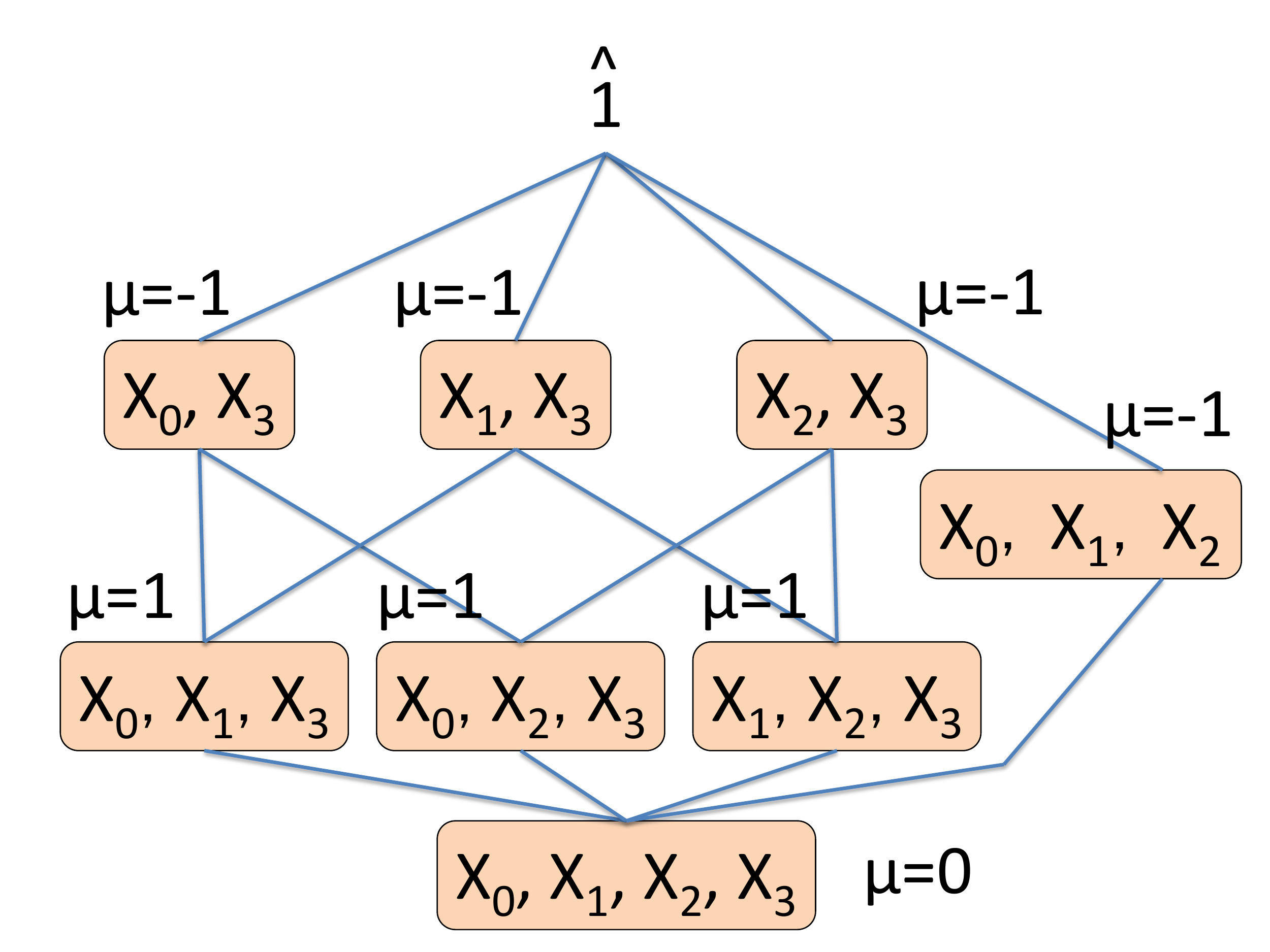}
%\hfill
%\centerline{\hfill (a) \hfill (b) \hfill}
\caption{The lattices for (a) $f_W$, (b) $f_9$. } %the formula $f_W$ (a) and for $f_9$ (b).}
  \label{fig:qw:q9}
  %\vspace{-6mm}
\end{figure} 
\begin{example}
Here we give examples of easy and hard queries:\\[-3ex]
\begin{itemize}
\setlength{\itemsep}{-0.3ex}
	\item  For a trivial example, %$H_k 
	$\query_k = \query_{k0} \vee \cdots \vee \query_{kk}$ has a
single clause, hence its lattice has exactly two elements $\hat 0$ and
$\hat 1$, and $\mu(\hat 0, \hat 1)=-1$, hence %$H_k$ 
$\query_k$ is \#P-hard.  

\item Two more interesting examples %are given below, 
for $k=3$:
{\small
\begin{align*}
  f_W = & (X_0 \vee X_2) \wedge (X_0 \vee X_3) \wedge (X_1 \vee X_3) \\
  f_9~ = & (X_0 \vee X_3) \wedge (X_1 \vee X_3) \wedge (X_2 \vee X_3) \wedge (X_0 \wedge X_1 \wedge X_2)
\end{align*}
}
Their lattices, shown in \autoref{fig:qw:q9}, satisfy $\mu(\hat 0,
\hat 1)= 0$, therefore weighted model counting for $Q_W = f_W(\query_{30},
\query_{31}, \query_{32}, \query_{33})$ and $Q_9 = f_9(\query_{30}, \query_{31}, \query_{32},
\query_{33})$ can be done in polynomial time.  For example, to compute the
probability of $Q_W$ we apply the inclusion/exclusion formula on the
query expression and get $\Pr[Q_W]$ = 
\begin{align*}
%&\Pr[Q_W]=\\ 
    &\phantom{-}\Pr[\query_{30} {\vee} \query_{32}] + \Pr[\query_{30} {\vee} \query_{33}] +
	\Pr[\query_{31} {\vee} \query_{33}] \\
&  - \Pr[\query_{30} {\vee} \query_{32} {\vee} \query_{33}] - \Pr[\query_{30} {\vee} \query_{31}
  {\vee} \query_{33}]\\ 
 & - \Pr[\query_{30} {\vee} \query_{31} {\vee} \query_{32} {\vee} \query_{33}] 
  + \Pr[\query_{30} {\vee} \query_{31} {\vee} \query_{32} {\vee} \query_{33}]
\end{align*}
%
		%\begin{align*}
			%\Pr[Q_W] = & \Pr[\query_{30} \vee \query_{32}] + \Pr[\query_{30} \vee \query_{33}] +  \Pr[\query_{31} \vee \query_{33}] \\
			%- & \Pr[\query_{30} \vee \query_{32} \vee \query_{33}] - \Pr[\query_{30} \vee \query_{31}
			%\vee \query_{33}] \\
				%& \ \ \ \ \ \ \ \ - \Pr[\query_{30} \vee \query_{31} \vee \query_{32} \vee \query_{33}] \\
			%+ & \Pr[\query_{30} \vee \query_{31} \vee \query_{32} \vee \query_{33}]
		%\end{align*}
%
While computing $\Pr[\query_{30} \vee \query_{31} \vee \query_{32} \vee \query_{33}]$ is
\#P-hard (because this query is $\query_3$), the two occurrences of this
term cancel out, and for all remaining terms one can compute the
probability in polynomial time in $n$ (since each misses at least one
term $\query_{30}, \query_{31}, \query_{32}, \query_{33}$).  
%To compute $\Pr[Q_9]$ one proceeds similarly.  
Thus, weighted model counting can be done in
polynomial time for $Q_W$ (similarly for $Q_9$), at the query expression
level.  
\end{itemize}
\end{example}

On the other hand, \autoref{thm:main} proves that, if we ground $Q_W$ or $Q_9$ first,
then any \decDNNF-based model counting algorithm will take exponential time on
the lineage. This leads to the main separation result of this paper:

\begin{theorem}[Main Result]\label{TH:MAIN-SEP}
 If $Q$ is a Boolean
 combination of the queries in $\hk$ that depends on all $k+1$
 queries in $\hk$ such that $\mu(\hat 0, \hat 1)=0$, 
 then weighted model counting for $Q$ can be done in time polynomial
 in $n$, whereas all existing \decDNNF-based propositional algorithms for model counting require 
 exponential time on the lineage. 
\end{theorem}

\section{Exponential lower bounds for all $\lin_k$}
\label{sec:h0h1}

In this section we prove \autoref{TH:1} which gives lower bounds on the
sizes of \FBDD s computing all $\lin_k$.
We find it convenient to prove these bounds assuming a natural property of
FBDDs.   We show that we can ensure this property with
only minimal change in \FBDD\ size, yielding our claimed lower bounds.

%\subsubsection*{The Unit Rule and Monotone Formulas}

Let $\Phi$ be a Boolean formula.  A {\em prime implicant}
(or {\em minterm}) of $\Phi$ is a term $T$
such that $T \Rightarrow \Phi$ and no proper subterm of $T$ implies $\Phi$.
If $T$ involves $k$ variables then we call it a {\em $k$-prime implicant}. 
$1$-prime implicants are also known as {\em \opis\ variables}.
For example, $X$ and $W$ are \opis\ in $X \vee YZ \vee YU \vee W$.  

The following definition is motivated by the {\em unit clause rule} in DPLL
algorithms which are primarily designed for satisfiability of CNF formulas.
If there is any clause consisting of a single variable or its negation (a unit
clause), then DPLL immediately sets such variables, one after another, since
their value is forced.   

\begin{definition} \label{def:unit:rule}
  Let $\F$ be an \FBDD\ for a Boolean function $\Phi$. Call a node $u$ in $\F$
  a \emph{unit node} if $\Phi_u$ has a \opi, and a
  \emph{decision node} otherwise.   
  We say that $\F$ follows the {\em unit rule} if for every unit node $u$
  the variable tested at $u$ is a unit.
\end{definition}

In the special case that $\Phi$ is a monotone formula, we can apply a 
transformation in order to convert any \FBDD\ $\F$ for $\Phi$ into one
that follows the unit rule and is not much larger than $\F$.  

For a variable $X$ of $\Phi$, define the {\em degree} of $X$ in $\Phi$ to be
the maximum over all partial assignments $\theta$ of the number of \opis\ of
$\Phi[\theta\cup\set{X{=}1}]$ that are not units of $\Phi[\theta]$. 
(If $\Phi$ is a DNF
formula then the degree of $X$ is at most the number of distinct variables
to co-occur in terms with $X$.)   Write $\Delta(\Phi)$ for the maximum degree
of any variable in $\Phi$.  In section~\ref{subsection:unitrule} we prove the
following:

\begin{lemma} \label{lemma:unit}
If $\Phi$ is a monotone formula with \FBDD\ $\F$ of size $N$,
then $\Phi$ has an \FBDD\ of size at most $\Delta(\Phi)\cdot N$
that follows the unit rule.
\end{lemma}

Since $\lin_k$ obviously has degree at most $n$ (for variables $R(i)$ and
$T(j)$), we obtain the following corollary.

\begin{corollary}
\label{lem:unit_hk} If $\lin_k$ has an \FBDD\ of size $N$, then $\lin_k$ has
an \FBDD\ of size at most $n N$ that follows the unit rule.
\end{corollary}

%An \FBDD\ that follows the unit rule has the following structure. 
%Consider any node $u$.
%If $\Phi_u$ has no \opis, then the node $u$ may test any variable.  If
%$\Phi_u$ has one or more \opis, then $u$ is part of a prime implicant
%branch (or PI-branch) $u = u_1, u_2, \ldots, u_p$ such that (1) $u_i$
%tests a prime implicant, (2) its 1-branch leads to 1, (3) its 0-branch
%leads to $u_{i+1}$, (4) it has a single parent.
%
%The motivation behind this definition comes from a standard heuristics
%deployed by the DPLL algorithm, which, recall, is primarily designed
%for CNF formulas: 
%if there is a clause consisting of a single variable
%$X$, called a {\em unit rule}, then use $X$ as the next variable to
%test.  Since our interest is in the DNF formulas $\lin_k$, the
%NOT QUITE.  THERE MIGHT BE MORE THAN ONE SUCH VARIABLE
%role of a unit rule is taken by a 1-prime implicant. 

Now \autoref{TH:1} is an immediate consequence of \autoref{lem:unit_hk}
together with the following lemma.

\begin{lemma}
\label{lem:theta1}
  Every \FBDD\ $\F$ for $\lin_k$ that follows the unit rule has size $\geq 2^{(n-1)}$.
\end{lemma}

The proof of \autoref{lem:theta1} follows using a general technique in which one
defines a notion of {\em admissible} paths in $\F$.
We will give such a definition and show that no two admissible paths in $\F$
can lead to the same node of
$\F$ since they must correspond to different subfunctions of $\lin_k$.
We will further show that every admissible path branches off from other
admissible paths at least $n-1$ times, guaranteeing that $\F$ must contain
a complete binary tree of distinct nodes of depth $n-1$ (in which edges may have
been stretched to partial paths).

For the remainder of this section we fix some \FBDD\ $\F$ for
$\lin_k$ that follows the unit rule.
Given a path $P$ in $\F$,
let $Row(P)$ be the set of $i \in [n]$ so that $P$
tests $R(i)$ at a decision node or there are $\ell$ and $j$ so that $P$ tests
$S_\ell (i, j)$ at a decision node; similarly, let $Col(P)$ be the set of
$j \in [n]$ which $P$ tests $T(j)$ at a decision node or there is some
$\ell$ and $i$ so that $P$ tests $S_\ell(i, j)$ at a decision node. 
Let $\mathcal{P}$ be the set of partial paths $P$ starting at the root and
ending at a (non-leaf) decision node so that both $|Row(P)|<n$ and
$|Col(P)|<n$ but any extension of $P$ has either $|Row(P)|=n$ or $|Col(P)|=n$.

The following is an easy observation.

\begin{lemma}
\label{lem:RTdetermined}
For all $k \geq 0$, if $P_1, P_2 \in \mathcal{P}$ with
$\lin_k [P_1] = \lin_k [P_2]$ then the two paths test the same set of $R$
and $T$ variables and must assign those tested the same values.
\end{lemma}

\begin{proof}
Suppose that there is some $R(i)$ such that $P_1$ assigns $R(i)$
value $b\in \{0,1\}$ that $P_2$ does not.   
$\lin_k[P_2]=\lin_k[P_1]$ does not depend on $R(i)$ so we can assume without
loss of generality that $P_2$ assigns $R(i)$ value $1-b$.
%By definition there must exist $(i_1,j_1)\in [n]^2$ such that 
%Consider $(i_1,j_1)\notin Row(P_1)\times Col(P_1)$.
Suppose without loss of generality that $P_1$ sets $R(i)$ to 1 and
$P_2$ sets $R(i)$ to 0. 

First consider the case $k=0$. 
Let $j_1\in [n]-Col(P_1)$.
Since $P_2$ sets $R(i)$ to 0, $\lin_0 [P_2]$ does not depend on $S (i, j_1)$
but $\lin_0 [P_1]$ sets neither $T(j_1)$ nor $S (i', j_1)$ for any $i'$,
so it does depend on $S (i, j_1)$, a contradiction.

Now suppose that $k\ge 1$.
Let $j_2\in [n]-Col(P_2)$.
As $\F$ follows the unit rule, this implies that $P_1$ sets
$S_1 (i, j ) = 0$ for all $j$, which in particular implies that $\lin_k[P_1]$,
and thus $\lin_k[P_2]$, does not depend on $S_1 (i, j_2)$.
However, since $j_2\notin Col(P_2)$, all terms of $\lin_{k\ell}$ for
$\ell\in [k]$ involving indices $(i',j_2)$ are unset for every $i'$,
which implies that $\lin_k[P_2]$ depends on $S_1(i,j_2)$, a contradiction. 
The case when the difference is $T(j)$ is analogous.
\end{proof}

We will first prove Lemma \ref{lem:theta1} for $k = 2m + 1$ odd.
The cases when $k>0$ is even as well as when $k = 0$ use almost identical
techniques; the proofs in these cases are given in the appendix.
%Appendix~\ref{sec:app-h0h1}.

%\newpage
\begin{definition}
Let $P$ be a partial path through $\F$ starting at the root. It is \emph{admissible} if for all $i, j$, it is consistent with one of the four following assignments:\\[-3ex]
\begin{enumerate}
\setlength{\itemsep}{-0.3ex}
\item
$R(i) = T(j) = 0$ and $S_\ell (i, j) = 0$ for all $\ell$ odd and $S_\ell(i, j) = 1$ for all $\ell$ even,
\item
$R(i) = T(j) = 1$ and $S_\ell (i, j) = 1$ for all $\ell$ even and $S_\ell(i, j) = 0$ for all $\ell$ odd,
\item
$R(i) = 0, T(j) = 1$ and $S_\ell (i,j) = 1$ for all $\ell$ even and $S_\ell (i, j) = 0$ for all $\ell$ odd, or
\item
$R(i) = 1, T(j) = 0$ and $S_\ell (i, j) = 1$ for all $\ell$ even and $S_\ell (i, j) = 0$ for all $\ell$ odd.
\end{enumerate}
$P$ is \emph{forbidden} if it is not admissible.
Let $\mathcal{A} \subset \mathcal{P}$ be the set of admissible paths in $\mathcal{P}$.
(See \autoref{tab:eg-pattern} for the case $k=5$).\\[-4ex]
\end{definition}
\begin{table}[t]\footnotesize
\begin{center}
\resizebox{0.45\textwidth}{!}{
    \begin{tabular}{|c| c | c | c | c | c | c |}
    \hline
    $R(i)$ & $S_1 (i, j)$ & $S_2 (i, j)$ & $S_3 (i,j)$ & $S_4 (i, j)$ & $S_5 (i, j)$ & $T(j)$ \\ \hline
   0 & 1 & 0 & 1 & 0 & 1 & 0 \\ \hline
    1 & 0 & 1 & 0 & 1 & 0 & 1 \\ \hline
    0 & 0 & 1 & 0 & 1 & 0 & 1 \\ \hline
    1& 0 & 1 & 0 & 1 & 0 & 0 \\ \hline
    \end{tabular}
    }
\end{center}
\vspace*{-2ex}
   \caption{The patterns for admissible paths for $k=5$.}
\vspace{-0.2in}
	\label{tab:eg-pattern}
\end{table}
\begin{lemma}
\label{lem:uniquepaths}
If $P_1, P_2 \in \mathcal{A}$ are distinct then
$\lin_k [P_1] \neq \lin_k [P_2]$.
\end{lemma}

\begin{proof}
Suppose $P_1, P_2 \in \mathcal{A}$ are distinct with
$\lin_k [P_1] = \lin_k [P_2] = F$.
Let $u$ be the first node at which $P_1$ and $P_2$ diverge,
and assume without loss of generality that $P_1$ takes the $0$-edge and
$P_2$ takes the $1$-edge.
Notice that $u$ must be a decision node.
By \autoref{lem:RTdetermined}, 
the node $u$ cannot test a $R(i)$ or $T(j)$ variable
so it must test $S_\ell(i, j)$ for some $i, j$.
Assume that $\ell$ is even (the case when $\ell$ is odd is symmetrical;
switch the roles of $P_1$ and $P_2$).
Then $F$ does not contain the prime implicant
$S_\ell (i, j) S_{\ell + 1} (i, j)$ and does not contain any units, but along
$P_2$ the variable $S_\ell (i, j) = 1$ , so $S_{\ell + 1} (i, j) = 0$ along $P_2$.
This implies that $F$ does not contain the prime implicant $S_{\ell + 1} (i, j) S_{\ell + 2} (i, j)$ 
but since $P_1$ cannot set $S_{\ell + 1} (i, j) = 0$ as otherwise it would be forbidden,
this implies that $P_1$ sets $S_{\ell + 2} (i, j) = 0$.
Inductively, we conclude that
$S_{\ell + 2p} (i, j)$ is set to zero on $P_1$ and $S_{\ell + 1 + 2p} (i, j)$
is set to zero on $P_2$ for all non-negative integers
$p \leq (k - \ell - 1) / 2$.
In particular, $S_k (i, j) = 0$ along $P_2$, so the prime implicant
$S_k (i, j) T(j)$ does not appear in $F$; as $F$ has no units, $T(j)$ must be
set to zero in $P_1$, as otherwise it would be forbidden.
Doing the same procedure but inducting downwards, we also conclude that
$R(i) = 0$ in $P_1$ and $S_1(i, j) = 0$ in $P_2$.
However, by \autoref{lem:RTdetermined}, this implies that $R(i) =T(j) = 0$ in
$P_2$, and since $S_1 (i, j) = S_k (i, j) = 0$ we conclude that $P_2$ is
forbidden, which is a contradiction.
\end{proof}

\begin{proof}[of Lemma \ref{lem:theta1}]
By Lemma \ref{lem:uniquepaths}, it suffices to count how many paths are in
$\mathcal{A}$, because each such path must correspond to a unique node in the
\FBDD.
We show that there are at least $2^{n - 1}$ such paths.
For any path $P \in \mathcal{A}$, call an assignment at a decision node $u$
along $P$ \emph{forced} if taking the opposite assignment would have resulted
in a forbidden path, and call the assignment \emph{unforced} otherwise.
We claim that there are at least $n - 1$ unforced assignments along any path
$P \in \mathcal{P}$.
Since some extension of $P$ either sets some variable in all rows or in all
columns, $P$ itself must have either $|Row(P)|=n-1$ or $|Col(P)|=n-1$.
Without loss of generality assume that $|Row (P)| = n - 1$.
Then the patterns of admissible paths ensure that, %it is readily checked that 
for each $i\in Row(P)$, the first decision
node $u$
along $P$ testing a variable either of the form $R(i)$ or $S_\ell (i, j)$ for
some $\ell$ and $j$ must be unforced (see \autoref{tab:eg-pattern}). So there must be at least $n - 1$
unforced assignments.

We now define an injection from $\{0, 1\}^{n - 1}$ to $\mathcal{A}$,
as follows: map each sequence of bits $(a_1, \ldots, a_{n - 1})$ to the
unique path $P \in \mathcal{A}$ that at its $i$-th unforced decision takes the
$a_i$-edge for $i \leq n - 1$, takes the $1$-edge at all unforced decisions
after its first $n - 1$ unforced decisions, makes all forced decisions as
required, and at each unit node takes the $0$ branch.
The existence of such an injection implies that
$|\mathcal{A}| \geq 2^{n - 1}$, as claimed.
\end{proof}

\subsection{Proof of Lemma~\ref{lemma:unit}}
\label{subsection:unitrule}

We begin with a simple property of monotone functions.
For a formula $\Phi$ let $\units(\Phi)$ denote the set of units in $\Phi$.
The following proposition will be useful because it implies that for monotone
formulas, setting units cannot create additional units.

\begin{proposition}
\label{proposition:monotone:units}
If $\Phi$ is a monotone function and $W$ is a variable in $\Phi$, then
$\units(\Phi[W{=}0])\subseteq \units(\Phi)$.
\end{proposition}

%\begin{definition}
%We say that an FBDD $\F$ is {\em non-redundant} if and only if whenever
%$\F$ has a node $v$ that tests a variable $X$, the function $\Phi_v$ computed
%at node $v$ depends on the variable $X$.
%\end{definition}
%
%It is easy to see that any FBDD can be computed by an equivalent smaller FBDD
%that is non-redundant.

Let $\F = (V, E)$ be an FBDD for a monotone formula $\Phi$,
where $V$ and $E$, respectively, denote the nodes and edges of $\F$.
%We assume that $\F$ is non-redundant without loss of generality. 
%For every node $v\in V$, let $\units(v)$ denote $\units(\Phi_v)$. 
For every edge $e=(u,v)\in E$, define $\units(e)=\units(\Phi_v)-\units(\Phi_u)$.
Observe that by \autoref{proposition:monotone:units}, any edge $e$ for which
$\units(e)$ is non-empty must be labeled 1 in $\F$.

Fix some canonical ordering $\pi$ on the variables of $\Phi$.
Define the following transformation on $\F$ to produce an \FBDD\ $\F'$ for
$\Phi$ that follows the unit rule:
The set of nodes $V'$ of $\F'$ is given by: 
\begin{align*}
V' = & V \cup \setof{(e,i)}{e=(u,v) \in E, u \in V,\ 1 \leq i \leq |\units(e)|} %\label{eq:vprime}
\end{align*}
The other details of $\F'$ are given as follows:
\begin{itemize}
\setlength{\itemsep}{-0.3ex}
\item For $e=(u,v)\in E$, the new vertices $(e,1),\ldots, (e,|\units(e)|)$
will appear in sequence on a path from $u$ to $v$ that replaces the edge $e$.
(If $\units(e)$ is empty then the original edge $e$ remains.)
\item Edge $(u,(e,1))$ in $\F'$ will have label 1, which is the label that
$e$ has in $\F$.
\item The variable labeling each new vertex $(e,i)$ in $V'$ will be the $i$-th
element of $\units(e)$ under the ordering $\pi$; we denote this variable by
$Z_{e,i}$.
\item The 1-edge out of each new vertex $(e,i)$ will lead to the 1-sink.  The
0-edge will lead to the next vertex on the newly created path.
\item For a vertex $w\in V$ labeled by a variable $W$, if $W$ appears in
$\units(e)$
for any edge $e=(u,v)$ such that there is a path in $\F$ from $v$ to $w$ then
the node $w$ becomes a no-op node in $\F'$, namely its labeling
variable $W$ is removed, its 1-outedge is removed, and its 0-outedge is
retained with no label.
Otherwise, $w$ keeps the variable label
$W$ as in $\F$ and its outedges remain the same in $\F'$.
\end{itemize}

The size bound required for \autoref{lemma:unit} is immediate by construction
since the degree of a variable upper-bounds the number of new units that
setting it can create.
However, in order for this construction to be well-defined we need to ensure
that the conversion to no-op nodes does not conflict with the conversion of
edges to paths of units.

\begin{proposition}
\label{prop:well-defined}
If the variable $W$ labeling $w$ is in $\units(e)$ for some edge $e=(u,v)$
for which there is a path from $v$ to $w$, then the outedges $e'$ of $w$
have $\units(e')=\emptyset$.
\end{proposition}

\begin{proof}
The assumption implies that $W$ is a unit of some $\Phi_{v}$.
Therefore $\Phi_{v}=W\lor \Phi'_{v}$ for some $\Phi'_{v}$.
Since $\F$ is an FBDD and $W$ labels $w$, $W$ is not set on the path from $v$
to $w$, hence $\Phi_w=W\lor \Phi''$ for some formula $\Phi''$. 
A 0-outedge $e_0$ from $w$ always has $\units(e_0)=\emptyset$ and
the 1-outedge $e_1=(w,w')$ of $w$ sets $W$ to 1 and hence $\Phi_{w'}=1$,
which implies that $\units(e_1)$ is also empty.   
\end{proof}

The following simple proposition is useful in reasoning about the correctness
of our construction.

\begin{proposition}
\label{prop:forever-unit}
If there is a path from $u$ to $v$ in $\F$ and $X \in \units(\Phi_u)$ then
either $X\in \units(\Phi_v)$, or $\Phi_v=1$, or $X$ is queried
on the path from $u$ to $v$ and hence $\Phi_v$ does not depend on $X$.
\end{proposition}

\begin{proof}
$X\in \units(\Phi_u)$ implies that $\Phi_u=X\lor F$ for some monotone formula
$F$.  If $X$ is set on the path from $u$ to $v$ then $\Phi_v$ does not 
depend on $X$; otherwise $\Phi_v=X\lor F'$ for some monotone formula $F'$
and either $X$ is a prime implicant or $F'$ is the constant 1 and hence
$\Phi_v=1$.
\end{proof}

Taken together with the size bound for our construction, the following lemma
immediately implies \autoref{lemma:unit}.

\begin{lemma}\label{lem:unit-correctness}
Let $\Phi$ be monotone and computed by FBDD $\F$.
Then $\F'$ is an FBDD for $\Phi$ that follows the unit rule.
\end{lemma}

\begin{proof}
We first show that $\F'$ is an FBDD, namely, every root-leaf path $P$ in $\F'$
queries each variable at most once.   
  $P$ contains old nodes $u\in V$ and new nodes $(e,i)$.  
  Suppose that a variable $X$ is tested twice along a path.  Clearly the two
  tests cannot be done by old nodes since $\F$ is an FBDD.  
  It cannot be tested by an old node $u$ and later by a new node
  $(e,i)$, because once tested by $u$, for any descendent node $v$, the
  formula $\Phi_v$ no longer depends on $X$, hence $X \not\in \units(e)$.
  It cannot be first tested by a new node $(e,i)$ and then later by an old
  node $u$ since the test at the old node would have been removed and
  converted to a no-op by the last item in the construction of $\F'$.
  Finally, suppose that the two tests are done by two new nodes $(e_1,i)$, and
  $(e_2,j)$ on $P$, where we write $e_1 = (u_1,v_1)$ and $e_2=(u_2,v_2)$. 
  then we must have $X \in \units(v_1)$ and $X \not\in\units(u_2)$
  where there is a path from $v_1$ to $u_2$ in $\F$.  
  By \autoref{prop:forever-unit}, this implies that $\Phi_{u_2}$ does not
  depend on $X$ which contradicts the requirement that
  $X\in \units(v_2)$ since $v_2$ is a child of $u_2$.

By construction, $\F'$ obviously follows the unit rule.
It remains to prove that $\F'$ computes $\Phi$.  We show something slightly
stronger:
For any function $F$, define $F^-$ to be
$F[\units(F){=}0]$, in which all variables in $\units(F)$ are set to 0.
We claim by induction that 
for all nodes of $v\in V$, if $\theta'$ labels a path in $\F'$ from the
root to $v$, then $\Psi[\theta']=\Phi_v^-$.
and $\theta'=\theta\cup \set{\units(\Phi_v)=0}$ for some $\theta$
that labels a path in $\F$ from the root to $v$.
This trivially is true for the root.
If it is true for the output nodes, then $\F'$ correctly computes $\Psi$
since constant functions have no units.
Let $v\in V$ and suppose that this is true for all vertices $u$ such that
there is some path $\theta'$ from the root to $v$ in $\F'$ for which
$u$ is the last vertex in $V$ on $\theta'$.
By the construction, for each such $u$ there must be an edge $e=(u,v)\in E$.
Suppose that the variable tested at $u$ in $\F$ is $W$.
We have 3 cases: If $e=(u,v)\in E$ is a 1-edge then $\Phi_v=\Phi_u[W{=}1]$.
Every path $\theta'$ from the root to $v$ 
through $u$ is of the form
$\theta'=\theta\cup\set{W{=}1}\cup \set{\units(e){=}0}$ for some $\theta$
that labels a path from the root to $u$ in $\F'$.  (This is true even if
$\units(e)$ is empty.)
By induction, $\Psi[\theta]=\Phi_u^-=\Phi_u[\units(\Phi_u){=}0]$ and
by definition $\units(\Phi_v)=\units(\Phi_u)\cup \units(e)$ so
$\Psi[\theta']=\Phi_u[W{=}1\cup \set{\units(\Phi_v){=}0}]=\Phi_v[\units(\Phi_v){=}0]$ as required.
If $e=(u,v)\in E$ is a 0-edge of $\F$ then $\Phi_v=\Phi_u[W{=}0]$.
If $u$ became a no-op vertex in $\F'$ then there was some ancestor $w$ of $u$
at which $W$ became a unit of $\Phi_w$.  Since $\F$ is an FBDD, it does not
query $W$ between that ancestor and $u$.
By \autoref{prop:forever-unit}, either $W\in\units(\Phi_u)$ or
$\Phi_u=1$.  In the latter subcase, $\Phi_v=\Phi_u=1$ and the correctness for
$u$ implies that for $v$.  In the former case,
$\units(\Phi_u)=\units(\Phi_v)\cup\set{W}$ and $\Phi_u^-=\Phi_v^-$ and
again the correctness for $\Phi_u^-$ implies that for $\Phi_v^-$.
In the case that $u$ does not become a no-op vertex, $W$ is not a unit of
$\Phi_u$ so $\units(\Phi_u)=\units(\Phi_v)$ and the fact that all paths
to $u$ yield $\Phi_u[\units(\Phi_u){=}0]=\Phi_u[\units(\Phi_v){=}0]$ for
all paths to $v$ that pass through $u$ as the previous vertex in $V$,
The last edge to $v$ adds the extra $W{=}0$ constraint.  Adding this
constraint to $\Phi_u[\units(\Phi_v){=}0]$ yields 
$\Phi_v[\units(\Phi_v){=}0]=\Phi_v^-$ as required. 
Therefore the statements holds for all possible paths from the root to $v$.
%The replacement of each edge $e=(u,v)$ by a path of unit tests clearly maintains
%this property inductively, namely if it is true at node $v$ then it remains
%true at node $u$.  
%The maintenance of the property through a node $w$ of $V$ labeled $W$ that is
%converted to a no-op node when $W$ is unit in some ancestor $v$ of $w$
%follows from the fact that $W$ is not queried prior to node $w$ and
%hence by \autoref{prop:forever-unit}, $W$ either
%is a unit of $\Phi_w$ or $\Phi_w=1$.  
%In the first case the no-op conversion exactly corresponds to
%setting the unit $W$ to 0 as required; in the second case, the function 
%computed by the 0-child also must be the constant 1 as required.
\end{proof}

\section{Lower bounds for Boolean combinations over $\Hk$}
\label{sec:hk}

In this section we prove \autoref{TH:3}.
Throughout this section we fix $f(\vecX)=f(X_0, \ldots, X_k)$, a Boolean
function that depends on all variables $\vecX$, and a domain size
$n>0$.  
To prove \autoref{TH:3}, we first prove that any \FBDD\ for
the lineage of the query $Q = f(h_{k0}, \ldots, h_{kk})$ can be
converted into a multi-output \FBDD\ for all of 
$\Hk=(\lin_{k0}, \lin_{k1}, \ldots, \lin_{kk})$
with at most an $O(k 2^k n^3)$ increase in size.  
The proof is constructive.  
\autoref{TH:3} then follows immediately using \autoref{TH:1} since
any \FBDD\ for $\Hk$ yields an \FBDD\ for $\lin_k$ of the same
size.

Recall that  $\lin_{k\ell}$ denotes the lineage of $\query_{k\ell}$ and
let $\Psi = f(\Hk)=f(\lin_{k0}, \ldots, \lin_{kk})$ be the lineage of $Q$. 
%The set of
%Boolean variables is
%
%$$\mathbf{Z} = \setof{R(i), S_1(i,j), \ldots, S_k(i,j), T(j)}{i,j \in
  %[n]}$$
%

If $\F$ is an \FBDD\ for $\Psi=f(\Hk)$, we will let $\Phi_u$ denote the
Boolean function computed at the node $u$; thus $\Psi = \Phi_r$,
where $r$ is the root node of $\F$. 
By the correctness of $\F$, all paths $P$ leading to $u$ have the property
that $\Psi[P]=\Phi_u$.
%MOVE TO SECTION 2
%Given a Boolean function $\Phi$, we let $Var(\Phi)$ denote the set
%of variables on which $\Phi$ depends.

In order to produce the \FBDD\ $\F'$ for $\Hk$ from $\F$ computing
$\Psi=f(\Hk)$, we would like to ensure that every internal node $v$ of
$\F'$ has the property that all paths $P$ leading to $v$ not only are
consistent with the same residual function $\Phi_v=\Psi[P]$, but they also all
agree on the residual values of $\lin_{k\ell}(v)\myeq \lin_{k\ell}[P]$ for all
$\ell$.
Since we will not easily be able to characterize its paths we find it
convenient to define this property not only with respect to
paths of $\F'$ but for formulas $\Phi_v$ with respect to arbitrary partial
assignments $\theta$.
We use the term {\em transparent} to describe this property since it ensures
that the value of $\Phi_v$ automatically also reveals the values
for all $\lin_{k\ell}(v)$.

\begin{definition} \label{def:transparent} Fix $\Psi = f(\lin_{k0},
  \ldots, \lin_{kk})$.  A formula $\Phi$ that is a restriction of $\Psi$
  is called {\em transparent}
%\footnote{We are interested only in formulas $\Phi$ of the
% form $\Psi[\theta]$ for some partial assignment $\theta$.  If
%$\Phi$ is not of this form, then it is technically considered to
%be transparent.}  
    if there exist $k+1$ formulas $\varphi_0, \ldots,
  \varphi_k$ such that, for every partial assignment $\theta$, if
  $\Phi = \Psi[\theta]$, then $\lin_{k0}[\theta] = \varphi_0$,
  $\ldots$, $\lin_{kk}[\theta] = \varphi_k$. We say that $\Phi$ {\em
    defines} $\varphi_0, \ldots, \varphi_k$.  
%    (b) Let $u$ be a node of
%some \FBDD\ $\F$.  We say that $u$ is {\em pathwise transparent}
%if the residual formula $\Phi_u$ has the property above holds when
%$\theta$ is restricted to the paths in $\F$ from the root to the
%node $u$, and we say $\F$ is pathwise transparent if all nodes in $\F$
%are pathwise transparent.
\end{definition}

In other words, assuming that $\Phi$ is derived from $\Psi$ as
$\Phi=\Psi[\theta]$
for some partial assignment $\theta$, then $\Phi$ is transparent if the
formulas $\lin_{k0}[\theta], \ldots, \lin_{kk}[\theta]$ are
uniquely defined; i.e., are independent of $\theta$.  Equivalently, for
any two assignments $\theta, \theta'$, if $\Psi[\theta] =
\Psi[\theta'] = \Phi$, then for all $0 \leq \ell \leq k$,
$\lin_{k\ell}[\theta] = \lin_{k\ell}[\theta']$.

\begin{example} \label{ex:transparent}
  Let $k=3$ and $f = X_0 \vee X_1 \vee X_2 \vee X_3$.  Given a domain
  size $n > 0$, the formula $\Psi$ is:
	{\small
	$$\bigvee_{i,j} R(i)S_1(i,j) \vee \bigvee_{i,j}  S_1(i,j)S_2(i,j) 
            \vee \bigvee_{i,j}  S_2(i,j)S_3(i,j) \vee \bigvee_{i,j} S_3(i,j)T(j)$$
	}					
    %\begin{align*}
    %\Psi = & \bigvee_{i,j} R(i)S_1(i,j) \vee \bigvee_{i,j}  S_1(i,j)S_2(i,j) \\
             %& \bigvee_{i,j}  S_2(i,j)S_3(i,j) \vee \bigvee_{i,j} S_3(i,j)T(j)
  %\end{align*}
  $\lin_{30}, \ldots, \lin_{33}$ denote each of the four
  disjunctions above.  Let $\Phi = R(3)S_1(3,7) \vee S_1(3,7)S_2(3,7)$.
   %\begin{align*}
     %\Phi = & R(3)S_1(3,7) \vee S_1(3,7)S_2(3,7)
   %\end{align*}
   There are many partial substitutions $\theta$ for which $\Phi =
   \Psi[\theta]$: for example, $\theta$ may set to 0 all variables
   with index $\neq (3,7)$, and also set $S_3(3,7)=T(7)=0$; or, it
   could set $S_3(3,7)=0$, $T(7)=1$; there are many more choices for
   variables with index $\neq (3,7)$.  However, one can check that,
   for any $\theta$ such that $\Phi=\Psi[\theta]$, we have:
   \begin{align*}
\lin_{30}[\theta]  = & R(3)S_1(3,7) & \lin_{31}[\theta] = & S_1(3,7)S_2(3,7) \\
\lin_{32}[\theta]  = & 0 & \lin_{33}[\theta] = & 0
   \end{align*}
   Therefore, $\Phi$ is {\em transparent}.  On the other hand, consider $\Phi'= S_1(3,7)$.
%
   %\begin{align*}
     %\Phi'= & S_1(3,7) 
   %\end{align*}
%
   This formula is no longer transparent, because it can be obtained by
   extending any $\theta$ that produces $\Phi$ with either $R(3)=0,
   S_2(3,7)=1$, or $R(3)=1, S_2(3,7)=0$, or $R(3)=S_2(3,7)=1$, and
   these lead to different residual formulas for $\lin_{30}$ and
   $\lin_{31}$ (namely $0$ and $S_1(3,7)$, or $S_1(3,7)$ and $0$, or
   $S_1(3,7)$ and $S_1(3,7)$).  
%   However, consider an \FBDD\ for $\Psi$,
%   and let $u$, $v$ be two nodes whose residual formulas are $\Phi_u =
%   \Phi$ and $\Phi_v = \Phi'$.  Suppose $v$ is only reachable from the
%   root by going through $u$, and that the only path from $u$ to $v$
%   corresponds to the assignment $R(3)=0$ and $S_2(3,7)=1$ (the other
%   two assignment leads to a different nodes).  Then $v$ is transparent
%   {\em relative to $\F$}.
\end{example}

In order to convert an \FBDD\ $\F$ for $\Psi=f(\Hk)$ into a multi-output
\FBDD\ for $\Hk=(\lin_{k0}, \ldots, \lin_{kk})$, we will try to 
modify it so that the formulas defined by the restrictions reaching its
nodes become transparent without much of an increase in the \FBDD\ size.
To do this we will add new intermediate nodes at which the formulas may not
be transparent but we will be able to reason about its computations based
on the nodes where the formulas are transparent.

Observe that if we know that $\Phi_v=\Psi[\theta]$ is transparent and
we have a small multi-output \FBDD\ $\F_\theta$ for $\Hk[\theta]$ then
we can simply append that small \FBDD\ at node $v$ to finish the job
and ignore what the original \FBDD\ did below $v$.
Intuitively, the reason that $\Hk$ and $\lin_k$ might not have such small FBDDs
is the tension between the $R(i)S_1(i,j)$ terms, which give a preference for
reading entries in row-major order and the $S_k(i,j)T(j)$ terms, which suggest
column-major order, together with the intermediate $S_\ell(i,j)S_{\ell+1}(i,j)$
terms that link these two conflicting preferences.   If all of those links
are broken, then it turns out that there is no conflict in the variable order
and the difficulty disappears.   
This motivates the following
definition which we will use to make this intuitive idea precise.

\begin{definition} \label{def:transversal}
  Let $\theta$ be a partial assignment to $Var(\Hk)$. 
  \begin{itemize}
  \setlength{\itemsep}{-0.3ex}
\item A {\em transversal} in $\theta$ is a pair of indices $(i,j)$
  such that $R(i)S_1(i,j)$ is a prime implicant of $\lin_{k0}[\theta]$,
  $S_k(i,j)T(j)$ is a prime implicant of $\lin_{kk}[\theta]$, and
  $S_\ell(i,j)S_{\ell+1}(i,j)$ is a prime implicant of $\lin_{k\ell}[\theta]$
  for all $\ell\in [k-1]$.
   \item Call two pairs of indices (or transversals) $(i_1,j_1), (i_2,j_2)$ {\em independent} if $i_1
   \neq i_2$ and $j_1 \neq j_2$.  
   \item A Boolean formula
  is called {\em transversal-free} if there {\em exists} a $\theta$
  such that $\Phi = \Psi[\theta]$ and $\theta$ has no transversals.
  \end{itemize}
\end{definition}

We now see that assignments without transversals, or even those with few
independent transversals, yield small \FBDD s.

\begin{lemma}
\label{prop:transversals}
  Let $\theta$ be a partial assignment to $Var(\Hk)$. 
  If $\theta$ has at most $t$ independent transversals then there exists a
  multi-output
  \FBDD\ for $(\lin_{k0}[\theta], \ldots, \lin_{kk}[\theta])$ of
  size $O(k2^{k+t} n^2)$.
\end{lemma}

\begin{proof}
  We first show that if $t = 0$ ($\theta$ has no transversals)
  then there exists a small OBDD that computes $\Hk[\theta]$.
  
  Let $G_\theta$ be the following undirected graph.  The nodes are the
  variables $Var(\Hk)$, and the edges are pairs of variables
  $(Z,Z')$ such that $ZZ'$ is a 2-prime implicant in
  $\lin_{k\ell}[\theta]$ for some $\ell$.  Since $\theta$ has no
  transversals, all nodes $R(i)$ are disconnected from all nodes
  $T(j)$.  In particular, there exists a partition $Var(\Hk) =
  \mathbf{Z}' \cup \mathbf{Z}''$ such that all $R(i)$'s are in
  $\mathbf{Z}'$, all $T(j)$'s are in $\mathbf{Z}''$ and every
  $\lin_{k\ell}[\theta]$ can written as $\varphi_\ell' \vee
  \varphi_\ell''$ where $Var(\varphi_\ell') \subset \mathbf{Z}'$ and
  $Var(\varphi_\ell'') \subset \mathbf{Z}''$; in particular,
  $\varphi_0''= \varphi_k'=0$.

  Define {\em row-major order} of the variables $Var(\Hk) - \set{T(1),
    \ldots, T(n)}$ by:
  \begin{align*}
    R(1),&S_1(1,1),\ldots,S_k(1,1),S_1(1,2)\ldots,S_k(1,n),\\
    R(2),&S_1(2,1),\ldots,S_k(2,1),S_1(2,2)\ldots,S_k(2,n),\\
    &\ldots \\
    R(n),&S_1(n,1),\ldots,S_k(n,1),S_1(n,2)\ldots,S_k(n,n)
  \end{align*}
  Let $\pi'$ be the restriction of the row-major order to the
  variables in $\mathbf{Z}'$.  Similarly, let $\pi''$ be the
  restriction to $\mathbf{Z}''$ of the corresponding column-major
  order of the variables that omits the $R(i)$'s, and places the $T(j)$
  before all variables $S_\ell(i,j)$.
  We build a multi-output \OBDD\ using the order $\pi=(\pi',\pi'')$
  for $(\lin_{k0}, \ldots, \lin_{kk})$.  
  In the first part using order $\pi'$ it will compute
  each $\varphi_\ell'$ term in parallel in width $O(2^k)$ and in the second
  part it will continue by including the additional terms from
  $\varphi_\ell''$ using order $\pi''$.  Observe that, except for the
  $R(i)S_1(i,j)$ terms, each of the variables in the 2-prime implicants in
  $\varphi_\ell'$ appear consecutively in $\pi'$.   Each level of the \OBDD\
  will have at most $2^{k+3}$ nodes, one
  for each tuple consisting of a vector of values of the partially computed
  values for the $k+1$ functions $\varphi_\ell'$, remembered value of $R(i)$,
  and remembered value of the immediately preceding queried variable.
  In the part using order $\pi''$, the remembered value of $T(j)$ is used
  instead of the remembered value of $R(i)$.
  The size of $\F'$ is $O(k2^k n^2)$ since there are $kn^2+2n$ variables
  in total in $Var(\Hk)$.
% 
% Indeed, the \OBDD\ for
% $\lin_{k0} = \bigvee_{ij}R(i),S_1(i,j)$ has width 3, the \OBDD\
% for $\lin_{k\ell} = \bigvee_{i,j} S_\ell(i,j),S_{\ell+1}(i,j)$ has
% width 2, and one can synthesize an multi-output \OBDD\ $\F'$ for all
% $k$ functions whose width is their product, i.e.  $3\cdot 2^{k-1}$;
% hence the size of $\F'$ is $O(k2^k n^2)$.
% From this, we immediately obtain
% a multi-output \OBDD\ for $(\varphi_0', \ldots, \varphi_{k-1}')$, by
% first setting all variables in $\mathbf{Z}''$ to 0 (since they don't
% occur in any $\varphi_\ell'$), then setting all variables tested by
% $\theta$ to their values in $\theta$.  Similarly, we can construct
% an \OBDD\ of width $3 \cdot 2^{k-1}$ for $(\varphi_1'', \ldots,
% \varphi_k'')$ by traversing the variables in $\mathbf{Z}''$ in the
% order $\Pi''$.  Since $\Pi', \Pi''$ are over disjoint sets of
% variables, we can combine the two \OBDD s into one \OBDD\ of width
% $3^2 \cdot 2^{2(k-1)}$ to compute $(0, \varphi_1', \ldots,
% \varphi_k') \vee (\varphi_0'', \ldots, \varphi_{k-1}'',0) =
% \lin_{k0}[\theta], \ldots, \lin_{kk}[\theta]$, which completes
% the proof.
% THE ABOVE DOES NOT QUITE WORK BECAUSE THE OBDD NEEDS TO PROPAGATE THE VALUE OF
% H_k0 EVEN THOUGH IT IS NOT CHANGED IN THE SECOND PART.   ITS REPLACEMENT IS
% A LITTLE WASTEFUL BUT MORE INTUITIVE.
  
  For general $t$, let $I$ and $J$ be the sets of rows and columns,
  respectively, of the transversals $(i,j)$ in $\theta$.   Since $\theta$ has
  at most $t$ independent transversals, the smaller of $I$ and $J$ has size
  at most $t$.  Suppose that this smaller set is $I$; the case when $J$ is
  smaller is analogous.  In this case, every transversal $(i,j)$ of $\theta$
  has $i\in I$.
  Notice that if we set all $R(i)$ variables with $i \in I$ 
  in an assignment $\theta'$ then the assignment $\theta \cup \theta'$ has no
  transversals, and thus, by the above construction, $\Hk[\theta']$ can be
  computed efficiently by a multi-output \OBDD. Therefore, construct the
  \FBDD\ which first
  exhaustively tests all possible settings of these at most $t$ variables 
  in a complete tree of depth $t$, then at each leaf node
  of the tree, attaches the \OBDD\ constructed above.
\end{proof}

A nice property of a single transversal for $\theta$ is that its existence
ensures that each $\lin_{k\ell}$ is a non-trivial function of its remaining
inputs; more transversals will in fact ensure that less about each
$\lin_{k\ell}$ disappears.   We will see the following: if there are at least some small
number of independent transversals for $\theta$ (three suffice), then we can use
the fact that $f$ depends on all inputs to ensure that
$\Psi[\theta]=f(\Hk)[\theta]$ will be transparent provided one additional
condition holds: there is no variable which
we can set to kill off all transversals in $\theta$ at once.

If we didn't have this additional condition, then the construction of $\F'$ for
$\Hk$ would be simple:  We would just use \autoref{prop:transversals} at all
nodes $v$ of $\F$ at which all assignments $\theta$ for which
$\Phi_v=\Psi[\theta]$ do not have enough transversals to ensure transparency
of $\Phi_v$.

Failure of the additional condition 
is somewhat reminiscent of the situation with setting units in
\autoref{sec:h0h1}:
This failure means that there is some variable we can set to kill off 
all transversals in $\theta$ at once, which by \autoref{prop:transversals}
means that along the branch corresponding to that setting one can get an
easy computation of $\Hk$ (not quite as simple as fixing the value of the
formula to 1 by setting units as in \autoref{sec:h0h1}, but still easy).
It is not hard to see, and implied by the proposition below, which is
easy to verify,
that the only way to kill off multiple independent transversals at once is to
set such a variable to 1.
By analogy we call such variables {\em $\Hk$-units}.  

\begin{proposition}
\label{lemma:monotonicity}
Let $\Phi=\Psi[\theta]$ for some $\theta$ with $t$ independent transversals
and $\theta'=\theta\cup\set{W{=}b}$ for $b\in\set{0,1}$.
The number of independent transversals in $\theta'$ 
is in $\{t-1,t\}$ if $b=0$ and is in $\{0,t-1,t\}$ if $b=1$.
%Let $e=(u,v)$ be an edge of $\F$.  Let $\Phi_u=\Psi[\theta]$
%for some $\theta$ with $t$ independent transversals and
%$\Phi_v=\Psi[\theta']$.  
%The number of independent transversals in $\theta'$ 
%is in $\{t-1,t\}$ if $e$ is a 0-edge and is in $\{0,t-1,t\}$ if $e$ is a
%1-edge.
\end{proposition}

\begin{definition}
We say that a variable $Z$ is an $\Hk$-{\em unit} for the formula
$\Phi$ if $\Phi[Z=1]$ is transversal-free but $\Phi$ is not. 
We let $\hkunits(\Phi)$ denote the
set of $\Hk$-units of $\Phi$, and we say that $\Phi$ is {\em $\Hk$-unit-free} if
$\hkunits(\Phi)=\emptyset$.   
\end{definition}

In %Appendix~\ref{subsection:transparent}, 
Section~\ref{subsection:transparent}, we will prove the following,
which makes our intuitive claim precise.

\begin{lemma} \label{lemma:transparent}
Let $\Psi=f(\Hk)$ where $f$ depends on all its inputs.  
Suppose that there exists a $\theta$ with at least $3$ independent transversals
such that $\Psi [\theta] = \Phi$.
If $\Phi$ is $\Hk$-unit-free then $\Phi$ is transparent.
\end{lemma}

We will still need to deal with the situation when $\Phi$ has any
$\Hk$-units along with multiple independent transversals.
Our strategy will be simple: whenever we encounter an edge in $\F$ on which an
$\Hk$-unit is created (possibly more than one at once) and the resulting
formula has sufficiently many transversals then, just as with the
unit rule, we immediately test
these $\Hk$-units, one at a time, each one after the previous one has been set
to 0 (since the branch where it is set to 1 has an easy computation remaining).

In order to analyze this strategy properly, it will be useful to
understand how $\Hk$-units can arise.
Observe, that if $\Phi= \Psi[\theta]$ and $Z$ is a unit
for some $\lin_{k\ell}[\theta]$, for $0 \leq \ell \leq k$,
then $Z$ is an $\Hk$-unit for $\Psi[\theta]$, because setting
$Z=1$ we ensure that $\lin_{k\ell}[\theta \cup \set{Z=1}]=1$, wiping out
all transversals.  The following lemma, which we prove in
Section~\ref{subsection:transparent}, shows a converse
of this statement
under the assumption that $\theta$ has at least 4 independent transversals.

\begin{lemma}
\label{lemma:unitwrthk}
  Let $\Psi=f(\Hk)$ where $f$ depends on all its inputs.
  If $\Phi = \Psi[\theta]$ for some partial assignment
  $\theta$ that has at least 4 independent transversals,
  then $\hkunits(\Phi) = \bigcup_{\ell\in \{0,\ldots,k\}}\units(\lin_{k\ell}[\theta])$.
\end{lemma}

Since a transversal $(i,j)$ requires that all elements
of $\Hk$ have 2-prime implicants rather than units on the terms involving
$(i,j)$, \autoref{lemma:unitwrthk} immediately implies the following:

\begin{corollary}
\label{cor:transversals}
  If $\Phi = \Psi[\theta]$ for some partial assignment $\theta$,
  then no $\Hk$-unit of $\Phi$ is in the prime implicants 
  indexed by any transversal of $\theta$.
\end{corollary}

Since the formulas in $\Hk$ are monotone, by \autoref{lemma:unitwrthk}
and \autoref{proposition:monotone:units} if
created by setting a variable to 1.
Hence, if $\Phi$ has at least 4 independent transversals, then
setting all $\Hk$-units in $\Phi$ to 0 in turn yields a formula that
still has at least 4 independent transversals (by \autoref{cor:transversals})
and is $\Hk$-unit-free (by \autoref{lemma:unitwrthk}), and hence transparent
(by~\autoref{lemma:transparent} and \autoref{proposition:monotone:units}).

We now describe the procedure for building a multi-output \FBDD\ $\F'$
computing $\Hk$: Start with the \FBDD\ $\F$ for 
$\Psi$ and let $V$ and $E$ be, respectively, the vertices and edges
of $\F$.   
Let $V_{4}\subseteq V$ be the set of nodes $v\in V$ such that 
$\Phi_v=\Psi[\theta]$ for some assignment $\theta$ that
has at least $4$ independent transversals.
By \autoref{lemma:monotonicity}, $V_4$ is closed under predecessors
(ancestors) in $\F$; let $E_4$ be the set of edges in $\F$ whose endpoints are
both in $V_4$.
The following is immediate from \autoref{lemma:monotonicity} and the
definition of $V_4$.

\begin{proposition}
\label{prop:v4}
If $v\in V_4$ but some child of $v$ is not in $V_4$ then either or both of the
following hold:
(i) there is an assignment
$\theta$ with precisely 4 independent transversals such that $\Phi_v=\Psi[\theta]$,
or (ii) the variable $Z$ tested at $v$ is in $\hkunits(\Phi_v)$ and the
$0$-child of $v$ is in $V_4$.
\end{proposition}

We will apply a similar construction to that of
Section~\ref{subsection:unitrule} to the subgraph of $\F$ on $V_4$.
For $e=(u,v)$, define $\hkunits(e)=\hkunits(\Phi_v)-\hkunits(\Phi_u)$ to be
the set of new $\Hk$-units created along edge $e$.
There are two differences from the argument in
Section~\ref{subsection:unitrule}:
(1) we will only apply the construction to edges in $E_4$ and will build
the rest of $\F'$ independently of $\F$, and (2) unlike setting ordinary units
to 1, in which the
corresponding \FBDD\ edges simply point to the 1-sink, each setting of an
$\Hk$-unit to 1 only guarantees that the resulting formula is
transversal-free; moreover the transversal-free formulas resulting from
different settings may be different.
The details are as follows (see \autoref{fig:HkConversion} in \autoref{subsection:transparent}):
%\autoref{subsection:hk-conversionfig}):
\begin{itemize}
{\setlength{\itemsep}{-0.3ex}
\item For every $e=(u,v)\in E_4$ such that $\hkunits(e)$ is non-empty
(and hence the 0-child of $u$ is also in $V_4$), add new
vertices $(e,1),\ldots,(e,|\hkunits(e)|)$ and replace $e$ with a path from
$u$ to $v$ having the new vertices in order as internal vertices.  
\item Edge $(u,(e,1))$ in $\F'$ will have label 1, which is the label that
$e$ has in $\F$; denote the variable tested at $u$ by $W$.
\item The variable labeling each new vertex $(e,i)$ will be the $i$-th
element of $\hkunits(e)$ under some fixed ordering of variables; we
denote this variable by $Z_{e,i}$.
\item The 0-edge out of each new vertex $(e,i)$ will lead to the next vertex
on the newly created path.  However, unlike the simple situation with ordinary
units, the 1-edge out of each new vertex $(e,i)$ will lead to a distinct
new node $(u,i)$ of $\F'$.   
Since $(u,v)\in E_4$ there is some partial assignment $\theta$ 
such that $\Phi_u=\Psi[\theta]$, $\Phi_v=\Psi[\theta,W{=}1]$, and
$\theta\cup \{W{=}1\}$ has at least 4 transversals; for definiteness we will pick
the lexicographically first such assignment.
Define the partial assignment 
\begin{align*}
\theta(u,i)&=\theta\,\cup\{W{=}1\}\cup\{\hkunits(\Phi_u){=}0\}\\
&\quad \cup\{Z_{e,1}{=}0,\ldots,Z_{e,i-1}{=}0,Z_{e,i}{=}1\},
\end{align*}
to be the assignment that sets all $\Hk$-units in $\Phi_u$ to 0 along with the
first $i-1$ of the $\Hk$-units created by setting $W$ to 1.
The sub-dag of $\F'$ rooted at $(u,i)$ will be the size $O(k2^k n^2)$
FBDD for $\Hk[\theta(u,i)]$ constructed in \autoref{prop:transversals}.
%, which we also view as computing $\Psi[\theta(u,i)]$ since the value of $\Psi$
%is determined by the value of $\Hk$.
\item For any node $w\in V_4$, whose 0-child is in $V_4$,
such that $w$ is labeled by a variable $W$ that
was an $\Hk$-unit of $\Phi_v$ for some ancestor $v$ of $w$, convert
$w$ to a no-op node pointing to its 0-child; that is, remove its variable
label and its 1-outedge and retain its 0-outedge with its labeling removed.
\item For any node $v\in V_4$ with a child that is not in $V_4$ and to which
the previous condition did not apply, let $\theta$ be a partial assignment
such that $\Phi_v=\Psi[\theta]$ and $\theta$ has precisely 4 independent
transversals,
as guaranteed by \autoref{prop:v4}, make $v$ the root of the size
$O(k2^kn^2)$ FBDD for $\Hk[\theta']$ constructed in~\autoref{prop:transversals}
where $\theta'=\theta\cup \{\hkunits(\Phi_v)=0\}$.
\item All other labeled edges of $\F$ between nodes of $E_4$ are included in
$\F'$.
}
\end{itemize}

The fact that this is well-defined follows similarly to~\autoref{prop:well-defined}.

\begin{lemma}\label{lem:hkunit-correctness}
$\F'$ as constructed above is a multi-output FBDD computing $\Hk$
that has size at most $O(k2^k n^3)$ times the size of $\F$.
\end{lemma}

\begin{proof}
We first analyze the size of $\F'$: As in the analysis for computing $\lin_k$,
some nodes $u$ have one added unit-setting path of length at most $n$ and 
each node on the path of at the extremities of $V_4$ has a new added FBDD of
size $O(k 2^k n^2)$ yielding only $O(k2^k n^3)$ new nodes per node of $\F$.
Also, the fact that $\F'$ is an FBDD follows similarly to the proof in
\autoref{lem:unit-correctness}.

%For every node $v\in V_4$ if $\Phi_v$ is the function computed in $\F$ at node
%$v$, then we show by induction that for every partial assignment $\theta'$
%For nodes $v\in V_4$ if $\Phi_v$ is the function computed in $\F$ at node
%We show by induction over nodes $v\in V_4$ the the function 
If $\Phi_v$ is the function computed in $\F$ at node $v$ for
all $v\in V_4$,
then we show by induction that for every partial assignment $\theta'$
reaching $v$ in $\F'$, $\Psi[\theta']=\Phi_v[\hkunits(\Phi_v){=}0]$
and $\theta'=\theta\cup\set{\hkunits(\Phi_v){=}0}$ for some partial assignment
$\theta$ such that $\Phi_v=\Psi[\theta]$.  It is trivially true of the root.
The argument is similar to that for \autoref{lem:unit-correctness}.

We now see why this is enough.
Since $v\in V_4$, $\Phi_v[\hkunits(\Phi_v){=}0]$ is $\Hk$-unit-free and
has at least 4 transversals, and so it is transparent by
\autoref{lemma:transparent}.
It remains to observe that (i) each multi-output FBDD attached directly to 
any node $v\in V_4$ used a restriction $\theta$ of $\Psi$ that would lead to
that node in $\F'$, which, because $\Psi[\theta]$ is transparent, implies that
its leaves correctly compute the values of $\Hk$, and (ii) the same holds
for the restriction leading to node $(u,i)$ with parent $(e,i)$, namely, the
restriction used to
build the multi-output FBDD consists of a restriction $\theta$ that in $\F'$
would reach node $u\in V_4$ and for which $\Psi[\theta]$ is transparent,
together with the 
assignment $\set{W{=}1}\cup \set{Z_{e,1}{=}0,\ldots,Z_{e,i-1}{=}0,Z_{e,i}{=}1}$
which follows the unique path from $u$ to $(u,i)$.  Again this implies that its
leaves correctly compute the values of $\Hk$.
\end{proof}

\subsection{Proofs of Lemmas~\ref{lemma:transparent} and~\ref{lemma:unitwrthk}}
\label{subsection:transparent}

\begin{figure}[t]
  \centering
		\includegraphics[scale=0.45]{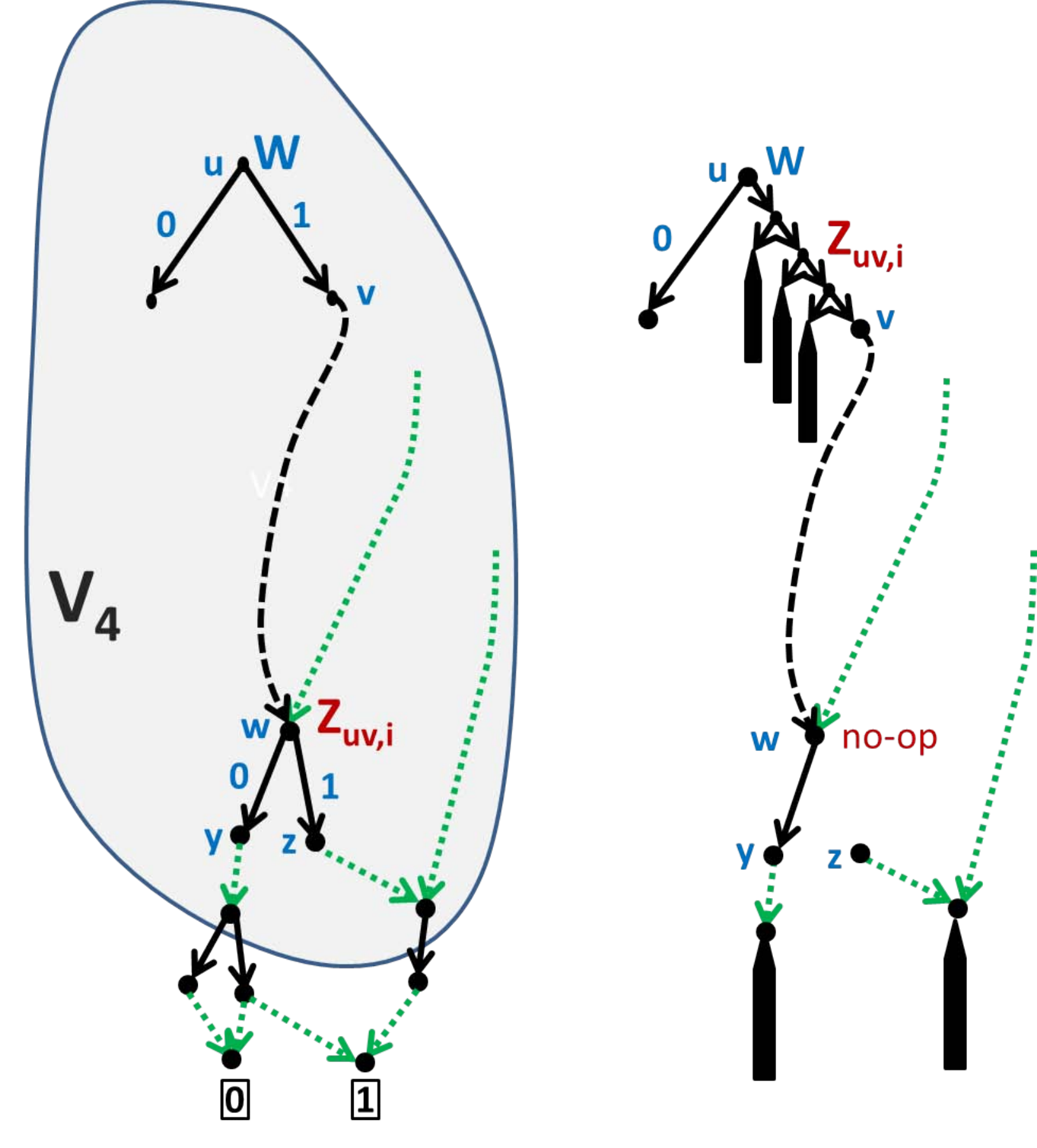}
    \vspace{-1ex}
    \centerline{(a)\hspace{1.5in}(b)}
			\includegraphics[scale=0.33]{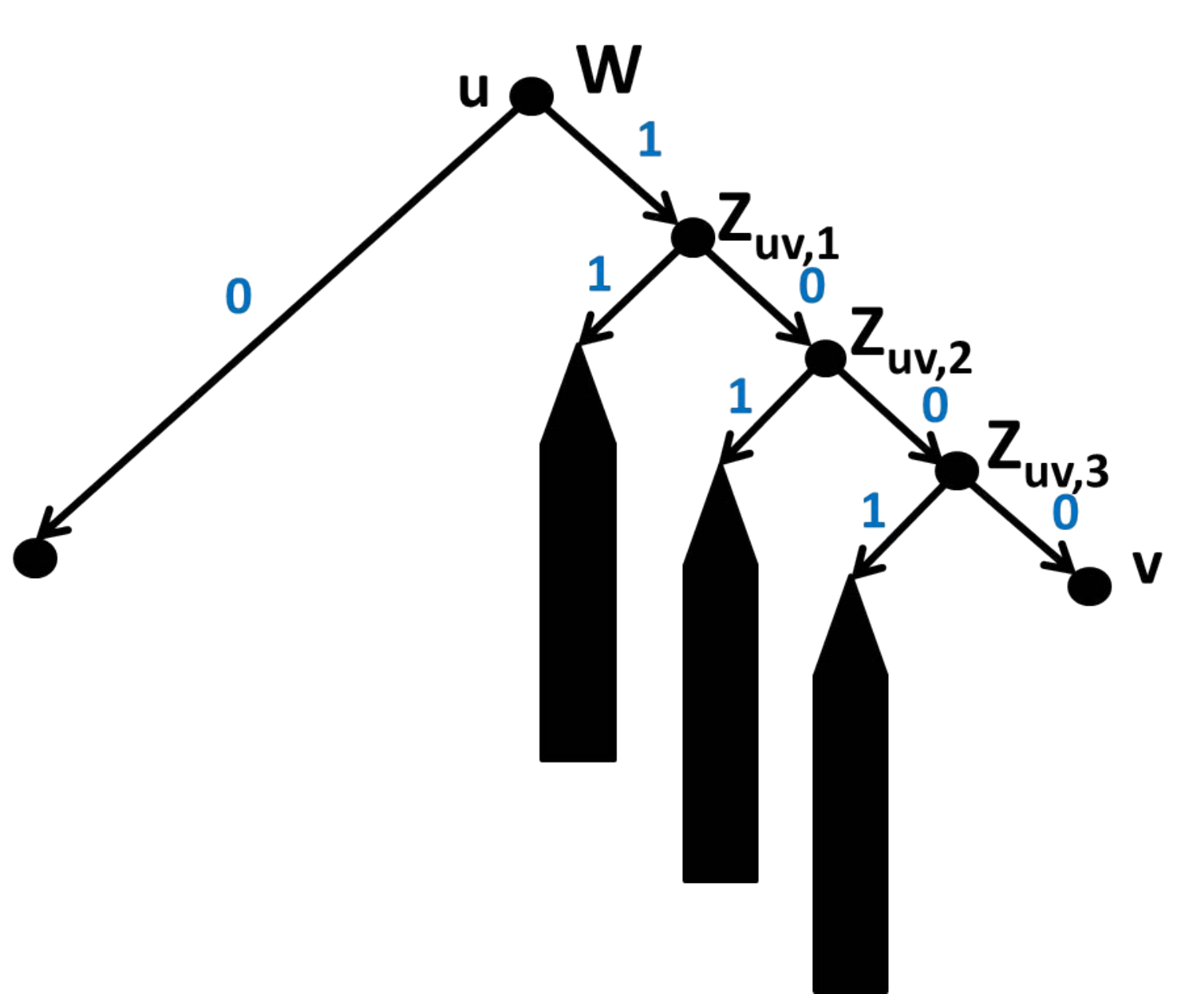}
    \vspace{-1ex}
    \centerline{(c)}
      \caption{Given an FBDD $\F$ for $\Psi=f(\Hk)$ in (a), apply the conversion
      to produce $\F'$ for $\Hk$ as in (b), with detail for unit propagation in (c) in case that setting $W=1$ produces new $\Hk$-units.}
        \label{fig:HkConversion}
	  \end{figure}
All formulas in $\Hk$ are 2-DNF formulas and, for every $(i,j)\in [n]^2$,
each has a unique 2-prime implicant $P_{\ell,i,j}$ indexed by $(i,j)$
where $P_{0,i,j}=R(i)S_1(i,j)$, $P_{k,i,j}=S_k(i,j)T(j)$ and 
$P_{\ell,i,j}=S_\ell(i,j)S_{\ell+1}(i,j)$ for $\ell\in [k-1]$.
We say that two of their 2-prime implicants, one from $\lin_{k\ell}$ and one
from $\lin_{k(\ell+1)}$, are {\em neighbors} if they share a variable
and hence have the same index $(i,j)$.
Observe that each prime implicant in $\lin_{k\ell}$ has two
neighbors if $\ell\in [k-1]$ and one neighbor if $\ell\in \set{0,k}$.

The key technical lemma is the following:

\begin{lemma} \label{prop:odd:even} 
Let $\Psi=f(\Hk)$ for some function $f$ that depends on all its inputs.
Suppose $\theta$ is a partial assignment with two independent
transversals $(i_0,j_0)$ and $(i_1,j_1)$. 
Suppose that for some $(i,j)$ independent of
of both transversals, the neighboring prime implicants of the
%$(i,j)$ 
prime implicant, $P_{\ell,i,j}$, of $\lin_{k\ell}$
are either unassigned or set to 0 by $\theta$.
Then there exists a partial assignment $\mu$ to all the variables of
$Var(\Psi[\theta])$, except those in $P_{\ell,i,j}$, and to
all variables of the transversals, $(i_0,j_0)$ and $(i_1,j_1)$,
such that $\Psi[\theta \cup \mu] = f_\ell(P_{\ell, i, j}[\theta])$.
Moreover, the choice of $\mu$ depends on $\Psi[\theta]$ (as well as 
the indices $\ell,i,i_0,i_1,j,j_0,j_1$) but not on any other aspect of $\theta$.
\end{lemma}

The lemma still holds when we merely assume that the three pairs of
indices are distinct; however we do not allow it here since it would
complicate the proof without any advantage with respect to our applications
of it.

\begin{proof}
Recall that since $f$ depends on all its inputs, for every $\ell$ there exists
an assignment $\mu_\ell : \vecX-\set{X_\ell} \rightarrow \set{0,1}$,
such that $f_\ell(X_\ell)\myeq f[\mu_\ell]$ is a function that depends on
$X_\ell$: i.e., $f_\ell(\vecX) = X_\ell$ or $f_\ell(\vecX)=\neg X_\ell$.  

We define assignment $\mu$ so that it sets the remaining variables to
force $\lin_{km}$ to equal $\mu_\ell(X_m)$ for all $m\ne \ell$ and force
$\lin_{k\ell}$ to equal $P_{\ell,i,j}[\theta]$.
In order to force some $\lin_{km}$ to 1 $\mu$ may need to set two
variables to 1 that may appear in neighboring prime implicants.
In order to avoid incidentally forcing any of those neighboring prime
implicants to 1, when forcing the $\lin_{km}$ to 1 we use the variables in the
two transversals alternately.   We now give the formal details.
%The condition on the neighboring prime implicants to $P_{\ell,i,j}$ ensures
%that ...

Let $Ones_\ell=\set{m\mid \mu_\ell(X_m)=1}$ and order the elements of
$Ones_\ell$ as $m_1<m_2<\cdots$, and define
$Ones^b_\ell=\set{m_r\mid b=r\mod 2}$ for $b\in \set{0,1}$.
For $b\in \set{0,1}$, define $\mu$ to set the variables of the $(i_b,j_b)$ prime
implicant of $\lin_{km}$ to 1 for every $m\in Ones^b_\ell$.
This will force $\lin_{km}[\theta\cup \mu]=1$, for all $m\in Ones_\ell$.
Let $\mu$ set all other variables in the transversals $(i_0,j_0)$ and
$(i_1,k_1)$ as well as all variables of $\Psi[\theta]$, except for those
in $P_{\ell,i,j}$, to 0.   
In particular, the alternation between how the 1's are forced in the definition
of $\mu$ ensures that for $b\in \set{0,1}$, if the $(i_b,j_b)$ prime
implicant of $\lin_{km}$ is set to 1 by $\mu$, then its neighboring prime
implicants are forced to 0 by $\mu$.
In fact,
$\lin_{km}[\theta\cup \mu]=0$ for all $m\notin  Ones_\ell\cup \set{\ell}$ since
each neighboring prime implicant to $P_{\ell,i,j}[\theta]$ will have one
variable set as in $\theta$ and the other set to 0 and all other prime
implicants are either set to 0 by $\theta$ or by $\mu$.
Finally, the same property is true of every prime implicant of $\lin_{k\ell}$
except for $P_{\ell,i,j}[\theta]$.
It remains to observe that $\Psi[\theta\cup\mu]=f(\Hk[\theta\cup\mu])=f[\mu_\ell](P_{\ell,i,j}[\theta])=f_\ell(P_{\ell,i,j}[\theta])$ and $\mu$ only depended
on $\theta$ through the value of $\Psi[\theta]$, as required.
\end{proof}

Notice that the lemma fails if $\theta$ has only 1 transversal:
\begin{example}
For a counterexample, consider $f(X_0, X_1, X_2) = X_0X_2 \vee X_1$.
Suppose that $\theta$ sets all variables in $\mathbf{Z}$ to 0, except for
the variables with indices $(i,j)=(3,7)$, which remain unset:
$$R(3), S_1(3,7), S_2(3,7), T(7).$$
Thus, $\theta$, has the transversal $(3,7)$.  However,
$\Psi[\theta]$ is:
\begin{eqnarray*}
  &  & f(\lin_{30}[\theta], \lin_{31}[\theta], \lin_{32}[\theta],  \lin_{33}[\theta]) \\
  & = &  f(R(3)S_1(3,7), S_1(3,7)S_2(3,7), S_2(3,7)T(7)) \\
  & = & R(3)S_1(3,7)S_2(3,7)T(7) \vee S_1(3,7)S_2(3,7) \\
  & = & S_1(3,7)S_2(3,7)
\end{eqnarray*}
hence it does not depend on $R(3)$ or $T(7)$. 
\end{example}

We immediately obtain the following two corollaries:

\begin{corollary} \label{corollary:depend} If $\theta$ has at
  least 3 distinct transversals 
  then all variables in its transversals are in
  $Var(\Psi[\theta])$.
\end{corollary}

\begin{proof}
This follows immediately since if $(i, j)$ is a transversal, all $P_{\ell, i, j}[\theta]$ are $2$-prime implicants in their respective $\lin_{k\ell} [\theta]$.
\end{proof}

\begin{corollary} \label{cor:same:transversals}
  If $\theta$ has at least 3 independent transversals
  then, every partial assignment $\theta'$ such that
  $\Psi[\theta]=\Psi[\theta']$ has the same set of transversals as
  $\theta$.
\end{corollary}

\begin{proof}
  We prove that every transversal of $\theta$ is a transversal of
  $\theta'$: this implies that $\theta'$ has at least 3 independent
  transversals, and therefore the converse holds too (every transversal
  of $\theta'$ is a transversal of $\theta$).  
  Let $\Phi=\Psi[\theta]=\Psi[\theta']$.
  Let $(i,j)$ be a transversal for $\theta$.  Since $\theta$ has at least 3
  transversals, by \autoref{corollary:depend}, $\Phi$ depends on all variables
  of the transversal $(i,j)$.  It follows that
  $\theta'$ cannot set any of these variables.  For $\ell\in [k-1]$, the
  2-prime implicants within each $\lin_{k\ell}$ are disjoint from each other
  and hence if the variables are unset then each such 2-prime implicant remains.
  Thus, for each $\ell\in [k-1]$ the Boolean function $\lin_{k\ell}[\theta']$
  contains the 2-prime implicant $S_{\ell}(i,j)S_{\ell+1}(i,j)$.

  It remains to prove that $\lin_{k0}[\theta']$ and $\lin_{kk}[\theta']$
  each contain the 2-prime implicants on $(i,j)$,   
  $R(i)S_1(i,j)$ or $S_k(i,j)T(j)$, which are unset by $\theta'$.
  To show this we must rule out $R(i)$ or $T(j)$ absorbing them. 
  We do this for $R(i)$; the case for $T(j)$ is analogous.
  Suppose to the contrary that $\lin_{k0}[\theta'\cup\set{R(i){=}1}]=1$. 
  If this is the case, then $\Psi[\theta'\cup\set{R(i){=}1}]=\Phi[R(i){=}1]$
  does not depend on any of the $R$ variables.
  However, since $\theta$ has $(i,j)$ as a transversal as well as, in
  particular, another (independent) transversal $(i',j')$ with $i'\ne i$,
  $\lin_{k0}[\theta]$ contains $R(i)S_1(i,j)$ and $R(i')S_1(i',j')$ as 
  2-prime implicants.  It follows that $\lin_{k0}[\theta\cup\set{R(i){=}1}]$
  depends on $R(i')$ and hence $\Psi[\theta\cup \set{R(i){=}1}]=\Phi[R(i){=}1]$
  depends on $R(i')$, contradicting our earlier derivation that it did not
  depend on any $R$ variables.
\end{proof}

Thus, for $k \geq 3$, the property of having $k$ independent transversals, 
is a property of the subformula and not of an assignment, and for the 
rest of the section we will say that a restriction $\Phi$ of $\Psi$ has
$k$ independent transversals if there exists an assignment $\theta$ so that
$\Phi = \Psi [\theta]$ and $\theta$ has $k$ independent transversals; by the
above, this
is equivalent to saying that all $\theta$ with $\Phi = \Psi [\theta]$
have $k$ independent transversals.

\subsection*{Proof of \autoref{lemma:transparent}}

We use the above to prove our lemma that are formulas are transparent if they
are unit-free and have sufficiently-many independent transversals.

\begin{proof}[of \autoref{lemma:transparent}]
  Suppose the contrary: then there exist two partial assignments
  $\theta, \theta'$ such that $\Phi = \Psi[\theta] =
  \Psi[\theta']$ and $\Phi$ has at least 3 independent traversals, but for
  some $\ell$, $\lin_{k\ell}[\theta] \neq
  \lin_{k\ell}[\theta']$.  
  Observe that if any $\lin_{k\ell}[\theta]$ or $\lin_{k\ell}[\theta']$
  were the constant 1, then $\Phi$ would be transversal-free 
  contradicting the fact that it has 3 independent traversals.
  Also observe that if any $\lin_{k\ell}[\theta]$ or $\lin_{k\ell}[\theta']$
  contained a 1-prime-implicant (unit) $Z$ then setting $Z{=}1$ would
  set the corresponding $\lin_{k\ell}$ to 1 which would eliminate
  all of its transversals, contradicting the assumption that $\Phi$ is
  $\Hk$-unit-free.  Therefore all prime implicants of $\lin_{k\ell}[\theta]$
  and $\lin_{k\ell}[\theta']$ are 2-prime implicants.
  Since all prime implicants in $\lin_{k\ell}[\theta]$ and
  $\lin_{k\ell}[\theta']$ are 2-prime implicants,
  \autoref{prop:odd:even} implies that all variables of all prime
  implicants are in $Var(\Phi)$.  

  Assume w.l.o.g.\ that $P_{\ell,i,j}$ is a 2-prime implicant of
  $\lin_{k\ell}[\theta]$ that is not a prime implicant of
  $\lin_{k\ell}[\theta']$; i.e., $P_{\ell,i,j}[\theta]=P_{\ell,i,j}$ but
  $P_{\ell,i,j}[\theta']=0$.
  Let $(i_0, j_0)$ and $(i_1,j_1)$ be two independent transversals for
  $\theta$ (which are also transversals for $\theta'$ by
  \autoref{cor:same:transversals}) that are also independent of $(i,j)$. 
  Since in both $\theta$ and $\theta'$, the neighboring implicants of
  $P_{\ell,i,j}$ either remain as 2-prime implicants or are
  set to 0 in $\theta$ and $\theta'$ (though not necessarily the same in both),
  we can apply \autoref{prop:odd:even} to both $\theta$ and
  $\theta'$ to obtain $\mu$ and $\mu'$.
  By the conclusion of \autoref{prop:odd:even},
  $\Phi[\mu]=\Psi[\theta\cup \mu]=f_\ell(P_{\ell,i,j})$ which is either
  $P_{\ell,i,j}$ or $\neg P_{\ell,i,j}$ but
  $\Phi[\mu']=\Psi[\theta'\cup \mu']=f_\ell(0)$ which is neither of the two.
  However, since the assignments $\mu$ and $\mu'$ depended only on the indices
  involved and $\Phi$, we conclude that $\mu = \mu'$ which is a contradiction.
\end{proof}

The requirement that $\Phi$ be unit-free is necessary for
\autoref{lemma:transparent} to hold: a simple example is given by the formula
$\Phi'=S_1(3,7)$ in \autoref{ex:transparent}, which is not transparent. 
The formula remains non-transparent even if we expand it with three
independent transversals,
e.g. $S_1(3,7) \vee [R(4)S_1(4,4) \vee \ldots \vee S_3(4,4)T(4)]
\vee [R(5)S_1(5,5) \vee \ldots] \vee [R(6)S_1(6,6) \vee\ldots]$, for
the same reasons given in the example.

\subsection*{Proof of \autoref{lemma:unitwrthk}}
Finally, we use the above properties to prove our characterization of
$\Hk$-units in case there are sufficiently many independent transversals.

\begin{proof}[of \autoref{lemma:unitwrthk}]
  If $Z\in \units(\lin_{k\ell}[\theta])$ for some $\ell$. then
  by definition, $\lin_{k\ell}[\theta \cup \set{Z{=}1}] = 1$, and therefore
  $\Psi[\theta \cup \set{Z{=}1}]=\Phi(\set{Z{=}1})$ is transversal-free; it
  follows that $Z\in \hkunits(\Phi)$.  \\
  Conversely, suppose that
  $Z \notin \bigcup_{\ell\in \{0,\ldots,k\}}\units(\lin_{k\ell}[\theta])$.
  In particular, none of the (monotone) formulas
  $\lin_{k\ell}[\theta\cup \set{Z{=}1}]$ is the constant 1.
  The assignment $Z{=}1$ can eliminate at most one of 
  the $\ge 4$ independent transversals in $\theta$, so
  $\theta\cup \set{Z{=1}}$ has at least 3 (independent) transversals.
%  WE CAN INCLUDE THE FOLLOWING EXAMPLE BUT IT ONLY SHOWS WHY independence
%  AMONG THE 4 TRANSVERSALS IS NEEDED.
%  (e.g, setting any $S_1(i,j){=}1$ will make $R(i)$ a unit
%  in $\lin_{k0}[\theta \cup \set{S_1(i,j){=}1}]$, wiping out all transversals
%  of the form $(i,j')$, but all transversals independent of $(i,j)$ remain in
%  $\theta \cup \set{S_1(i,j){=}1}$). 
  Therefore, by \autoref{cor:same:transversals},
  \emph{every} partial assignment $\theta'$ such that
  $\Psi[\theta']=\Phi[Z{=}1]=\Psi[\theta\cup \set{Z{=}1}]$ has at least
  3 independent transversals.   This implies that $\Phi[Z{=}1]$ is not
  transversal-free and hence $Z\notin \hkunits(\Phi)$.
\end{proof}

\section{A Dichotomy Theorem for Efficient Propositional Model Counting}
\label{sec:mini:dichotomy}

%\sudeepa{before Grounded vs Lifted?}
In this section we present a characterization for a restricted class of queries
%for the existence of (current) model counting algorithms on the propositional formulas. 
for the existence of efficient (current) model counting algorithms on the propositional formulas. 
%As a 
%We show that %it here for a very restricted class of queries.  
%for a restricted class of queries, either %there exists a \decDNNFn of 
For this class, either all \decDNNFn s require exponential size
(therefore all modern model counting algorithms take exponential time), or
we can construct a poly-size FBDD in polynomial-time (data complexity), leading to a polynomial-time
model counting algorithm.
%\footnote{In \cite{DBLP:journals/mst/JhaS13} a sufficient condition
%was given under which a UCQ is guaranteed to have a polynomial size
%\FBDD.}. 

\par
To define the class of queries, we 
consider the $k+1$ queries in $\hk$ defined in \autoref{sec:main-results}, 
%set of queries $\HH_k$ defined above, 
and the following additional $k+2$ queries $\gk$:
$\newquery_0 = \exists u_0 R(u_0)$, 
	$\newquery_{\ell} =    \exists u_{\ell} \exists v_{\ell} S_{\ell}(u_{\ell},v_{\ell})$, $\forall \ell \in [k]$, 
	$\newquery_{k+1} = \exists v_{k+1} T(v_{k+1})$.
These queries have the following lineages 
on the same domain of size $n$:
		$\newlin_{0} = \bigvee_{i\in [n]} R(i)$,  
			$\newlin_{\ell} = \bigvee_{i,j \in [n]} S_{\ell}(i,j)$, $\forall \ell \in [k]$,
			and $\newlin_{k+1} = \bigvee_{j \in [n]} T(j)$ respectively. 
Note that we consider data complexity, so we assume that the query is fixed (\ie, $k$ a is constant).

		%\begin{align*}
			%&\newquery_0 = \exists u_0 R(u_0) & \newquery_{k+1} = \exists v_{k+1} T(v_{k+1})\\
			%&\newquery_{\ell} =    \exists u_{\ell} \exists v_{\ell} S_{\ell}(u_{\ell},v_{\ell}) \quad \forall \ell \in [k]  
		%\end{align*}

		%having lineages on the same domain of size $n$:
		%\begin{align*}
			%\newlin_{0} = & \bigvee_{i\in [n]} R(i), \qquad & \newlin_{k+1} = \bigvee_{j \in [n]} T(j)\\
			%\newlin_{\ell} = & \bigvee_{i,j \in [n]} S_{\ell}(i,j)~~~~ \forall \ell \in [k] &
			%\end{align*}

Let $g(X_0, \cdots, X_k, Y_0, \cdots, Y_{k+1})$ be a Boolean function
on $2k+3$ variables. Consider the query $Q = g(\query_{k0}, \cdots,
\query_{kk}, \newquery_0, \cdots, \newquery_{k+1}) = g(\hk, \gk)$, and its lineage 
$g(\lin_k, \newlin_k)$. Let $g(\vecX, \mathbf{1}) = g(X_0, \cdots, X_k, 1, \cdots,  1)$. Then the following dichotomy holds
where $n$ is the size of the domain:
 \begin{theorem} \label{TH:4} 
\begin{enumerate}[(a)]
\item \label{item:mini-a} If %$g(X_0, \cdots, X_k, 1, \cdots,  1)$ 
$g(\vecX, \mathbf{1})$ depends on all $k+1$ variables $X_0, \cdots, X_k$, then any \decDNNFn\
  for the lineage of $Q$ has size $2^{\Omega(n)}$.  
	\item \label{item:mini-b} Otherwise, there exists an
  \FBDD\ for the lineage of $Q$ of size $n^{O(1)}$, and the \FBDD\ can be constructed in $n^{O(1)}$ time.
	\end{enumerate}
\end{theorem}
\noindent
%The proof is in \autoref{sec:mini:dichotomy}.
The first part of Theorem~\ref{TH:4} extends \autoref{TH:3}, where $f(\vecX) = g(\vecX, \mathbf{1})$.
%, where we had a Boolean function
%$F$ that does not use the variables $Y_0, Y_1, \cdots, Y_{k+1}$:
%therefore $F(X_0, \cdots, X_k, 1, \cdots, 1)$ is the same as $F$.
\begin{example}
We illustrate with three examples
%; we drop all existential quantifiers
%to avoid clutter.
\begin{itemize}
\item $g = (X_0 \vee B_{2}) \wedge (B_0 \vee X_1)$ where $k = 1$.
Then $g(\vecX, \mathbf{1}) = 1$: it does not depend on $X_0, X_1$, and therefore,
the lineage has a poly-size \FBDD. %an FBDD of polynomial size.
\item $g = X_0 \wedge (X_1 \vee B_3)
  \wedge (X_1 \vee B_5)
  \wedge (X_2 \vee X_3\vee X_4 \vee X_5)$ where $k = 5$.
Then $g(\vecX, \mathbf{1}) =  X_0 \wedge (X_2 \vee X_3\vee X_4 \vee X_5)$: it does not depend on $X_1$, and 
the lineage has a poly-size \FBDD\footnote{This is $Q_V$ in \cite{DBLP:journals/mst/JhaS13} for which a poly-size FBDD was shown.}. %n FBDD of polynomial size.
\item $g = (X_0 \vee X_1) \wedge (X_1 \vee B_3) \wedge (X_2 \vee X_3)$ where $k = 3$.
Then $g(\vecX, \mathbf{1}) =  (X_0 \vee X_1) \wedge (X_2 \vee X_3)$: it depends on all of $X_0, \cdots, X_3$, and 
therefore every \decDNNFn\ for the lineage has exponential size.  
\end{itemize}
\end{example}

In prior work~\cite{DBLP:journals/mst/JhaS13} a sufficient condition
was given under which a UCQ is guaranteed to have a polynomial size
\FBDD. 
	%; in particular it was shown that the query $Q_V$ constructed by the first example above 
	%has a polynomial size \FBDD.  
	Our result here is novel in that it represents a necessary and
sufficient condition, albeit for a very restricted fragment of UCQ.

%\dan{The proof is very simple, but the write-up can be improved.  I'm
  %open to alternative suggestions about this result.  We could: (1)
  %remove it complete from the paper, or (2) move the proof (i.e. this
  %section) to the appendix, or (3) leave it as is.}
%For one direction, 

Next we prove \autoref{TH:4}. 

\begin{proof}

\subsection*{Proof of (\ref{item:mini-a})}
Suppose $f(X_0, \ldots, X_k)$ $= g(\vecX, \mathbf{1})$ $= g(X_0,
\ldots, X_k, 1, \ldots, 1)$ depends on all variables $X_0, \ldots,
X_k$.  Let $\F$ be an \FBDD\ for the query $Q = g(\query_{k0},\ldots,
\query_{kk}, \newquery_0, \ldots, \newquery_{k+1})$, over the domain of size $n' = n+2$.  
We will convert $\F$ to a an \FBDD\ $\F'$ for query $Q' = f(\query_{k0},\ldots,
\query_{kk})$ over the domain of size $n$; further $\F$ and $\F'$ will have the same size.
By \autoref{TH:3}, the size of $\F'$ is $2^{\Omega(n)}$; therefore the size of $\F'$
is $2^{\Omega(n)}$ = $2^{\Omega(n')}$.
\par
To convert $\F$ to $\F'$, modify
$\F$ by setting the following values\footnote{This means: replace a
  node testing one of the variables $Z$ mentioned above by a no-op
  node, whose unique child is either the 0-child, or the 1-child of
  $Z$, according to the assignment.}, where $n_1 = n+1$ and $n_2 =
n+2$:
\begin{eqnarray*}
R(n_1) = 1, R(n_2) = 0 & \\
S_\ell(n_2,n_2) = 1 & \qquad \mbox{if $\ell$ is odd}\\
S_\ell(n_1,n_1) = 1 & \qquad  \mbox{if $\ell$ is even}\\
S_\ell(i,j) = 0 & \qquad \forall \ell \in [k]\\
\qquad\qquad & \mbox{$\forall$ other $(i, j) \in \set{n_1,n_2}\times \set{n_1,n_2}$}\\
T(n_1) = 1, T(n_2) = 0  & \qquad  \mbox{if $k$ is odd}\\
T(n_2) = 1, T(n_1) = 0  & \qquad \mbox{if $k$ is even}
%R(n_1) = S_1(n_2,n_2) = S_2(n_1,n_1) = S_3(n_2,n_2) = \ldots = 1 \\
%%T(n_2) = S_{k}(n_1,n_1) = S_2(n_2,n_2) = S_3(n_1,n_1) = \ldots = 1 \\
%S_\ell(i,j) = 0 \ \ \ \mbox{for all other $i \in \set{n_1,n_2}\vee j \in \set{n_1,n_2}$}
\end{eqnarray*}
%$T(n_1) = 1$ (resp. 0) and if $k$ is odd (resp. even).
% Further set $S_\ell(i,j)=0$ for all $i,j$ where $i\in\set{n_1,n_2}$
% or $j \in \set{n_1,n_2}$ except for the values above.
Note that he modified \FBDD\ $\F'$ computes the query $Q' = f(\query_{k0},
\ldots, \query_{kk})$ over a domain of size $n$:  all queries $\newquery_0,
\ldots, \newquery_{k+1}$ become \texttt{true} under the partial assignment
above, while the lineages for $\query_{k0}, \ldots, \query_{kk}$ over domain $n' = n+2$
becomes their lineage over the domain of size $n$.  
%It follows that
%the size of $Q$ is bounded from below by \autoref{TH:3}.

\subsection*{Proof of (\ref{item:mini-b})}
For the converse, assume that $f(X_0, \ldots, X_k) = g(X_0, \ldots, X_k, 1, \ldots, 1)$ does
not depend on $X_s$.  Denote $Q' = f(\query_{k0}, \ldots, \query_{kk})$: its lineage is \emph{transversal-free}
(\autoref{def:transversal}) and therefore it has a shared \OBDD\ $\F_s$ of
size $O(n)$ for formulas $\query_{k0}, \ldots, \query_{kk}$
(see  the proof of \autoref{prop:transversals} which constructs the \FBDD\ in poly-time for a fixed $k$).  
\par
Further, if for any $\ell$, $\newquery_\ell=0$,
then both $\query_{\ell-1} = \query_\ell = 0$, hence the lineage of the residual
formula $f(X_0, \ldots, X_k)$ is transversal free.  Therefore, the residual formula has a
shared \OBDD $\F_{0, \ell}$ of size $O(n)$ for $\query_{k0}, \ldots, \query_{kk}$ obtained by traversing the variables
$R(i),S_1(i,j), \ldots, S_{\ell-1}(i,j)$ in row-major order, and
traversing the variables $S_{\ell+1}(i,j), \ldots, T(j)$ in
column-major order. \\

We now describe the \FBDD\ $\F$
for $Q = g(\query_{k0}, \ldots, \query_{kk}, \newquery_0, \ldots, \newquery_{k+1})$
which will have $k+3$ layers: 0 to $k+2$.
(1) Layers 0 to $k+1$ of $\F$ will have a tree structure; 
%starts by first testing 
the $k+2$ queries $\newquery_0, \ldots, \newquery_{k+1}$ are tested in these layers one after one.  (As
an optimization, the \FBDD\ only needs to test those queries on which
the function $F$ depends.)  
%This part of the \FBDD\ is a tree,
%organized into $k+2$ levels for $\newquery_0, \ldots, \newquery_{k+1}$:
(2) In the $k+2$-th layer, there will be copies of \FBDD s $\F_s$ or $\F_{0, \ell}$, $\ell \in [k]$
described above.

The layers of $\F$ are described below:
\begin{description}
\item[Layer 0: Test $\newquery_0$] Test the variables $R(1), R(2), \ldots, R(n)$ in an arbitrary order: for
  each node $R(i)$, its 0-child is $R(i+1)$ and its 1-child is a root
  of a unique subtree at the next layer.  The 0-child of the last node $R(n)$
  is also leads to a unique subtree at the next layer.  The total number of
  edges to the next layer is $n+1$.
\item[Layer $\ell$, $1 \leq \ell \leq k$: Test for $\newquery_{\ell}$] 
Each sub-tree in the layer $\ell$ tests the variables $S_\ell(1,1), \ldots,
  S_\ell(n,n)$ in an arbitrary order (\eg, row-major): for each node
  $S_\ell(i,j)$ its 0-child is the next variable in this order, and
  each 1-child is a root of a unique subtree at the next layer.  The 0-child
  of the last node in this order is also a unique subtree at the next layer.
  The total number of edges from each of these subtrees to the next layer is $n^2+1$.
\item[Layer $k+1$: Test for $\newquery_{k+1}$] Test the variables $T(1), \ldots, T(n)$  in an arbitrary order: the
  0-child of $T(i)$ is $T(i+1)$.  All 1-children plus the 0-child of
  the last node %form the $n+1$ outputs.
	points to a unique \FBDD\ ($\F_s$ or $\F_{0, \ell}$ for some $\ell in [k]$) in the last layer.
\end{description}

Before we describe the last $k+2$-th layer,  we note that each of the outputs from the $k+1$-th layer encode two pieces of
information: (i) the values of all queries $\newquery_0, \ldots, \newquery_k$,
i.e. for each of them we know if it is 0 or 1, and (ii) we know which variables %$S_\ell(i,j)$ 
have been tested. In the \FBDD s $\F_s$ or $\F_{0, \ell}$s in the last layer, we set the values of these variables 
according to the test earlier (replace the variable by a no-op node having a unique child based on its 0- or 1-value)
to ensure that every variable is tested at most once in $\F$ (as is the case with the first $k+2$ layers). 

\begin{description}
\item[Layer $k+2$: Shared \FBDD s for $\query_{k0}, \ldots, \query_{kk}$] 
Consider any edge to layer $k+2$ from layer $k+1$.
(1) If for any $\ell$, $\newquery_\ell=0$, then both $query_{\ell-1} = \query_\ell=0$.
Consider the least $\ell$ such that $\query_{\ell} = 0$ and connect this edge to the shared \FBDD\ $\F_{0, \ell}$
substituting the values of the variables that have already been tested.
\par
(2) If for all $\ell$, $\newquery_\ell=1$, then connect the edge 
to the shared \FBDD\ $\F_{s}$ substituting the values of the variables that have already been tested.
\end{description}

\textbf{Computing the function $g(\query_0, \cdots, \query_k, \newquery_0, \cdots, \newquery_{k+1})$.~~} 
Now at the sinks of the \FBDD\ $\F$ (sinks of the \FBDD s in $k+2$-th layer), we know the values of the lineages for all queries
$\newquery_0, \cdots, \newquery_{k+1}$ (from 0 to $k+1$-th layers)  as well as for the queries $\query_0, \cdots, \query_k$ (from the $k+2$-th layer).
Therefore, the value of the lineage of the query $Q = g(\query_0, \cdots, \query_k, \newquery_0, \cdots, \newquery_{k+1}$ can be easily computed.\\

The total number of nodes in the %first $k+2$ layers is 
\FBDD\ $\F$ is
%$(n+1)^2(n^2+1)^{k-1} = n^{O(1)}$.
$O(n) \times [O(n^2)]^k \times O(n) \times O(n)$ = $n^{O(1)}$.
The completes the proof of part (\ref{item:mini-b}) of the theorem.

%Before we describe the last $k+2$-th layer,  we note that each of the outputs from the $k+1$-th layer encode three pieces of
%information: (a) the values of all queries $\newquery_0, \ldots, \newquery_k$,
%i.e. for each of them we know if it is 0 or 1, and (b) if $\newquery_\ell=0$,
%then we know that $S_\ell(1,1) = \ldots = S_\ell(n,n)=0$. (c) if
%$\newquery_\ell=1$ then we know which variables $S_\ell(i,j)$ have been tested
%(namely a certain prefix of the order in which we tested them): all
%these variables are 0, except for exactly one which is
%$S_\ell(i,j)=1$; all other variables are untested.
%
%Consider now one of these outputs.  If for any $\ell$, $\newquery_\ell=0$,
%then both $query_{\ell-1} = \query_\ell=0$, hence the lineage of the residual
%formula is transversal free.  Therefore, the residual formula has an
%\OBDD\ of size $O(n)$ obtained by traversing the variables
%$R(i),S_1(i,j), \ldots, S_{\ell-1}(i,j)$ in row-major order, and
%traversing the variables $S_{\ell+1}(i,j), \ldots, T(j)$ in
%column-major order.  If for all $\ell$, $\newquery_\ell=1$, then the residual
%formula is also transversal-free, because of our assumption that
%$F(X_0, \ldots, X_k, 1, \ldots, 1)$ does not depend on some $X_\ell$.\\
\end{proof}
%\section{Proof of Theorem~\ref{TH:2}}\label{app:proof-thm-decdnnfn}
\section{A Quasi-polynomial Conversion from DLDDs to FBDDs}\label{sec:proof-thm-decdnnfn}
Here we prove Theorem~\ref{TH:2}, that any \decDNNFn\ $\D$ with $N$ nodes
  can be converted into an equivalent \FBDD\ $\F$ computing the same formula as
  $\D$, where $\F$ has at most $N 2^{\log^2 N}$ nodes.
A \decDNNFn\ can have these types of nodes: (i) decision-nodes and (ii) 0- and 1-sinks as in \FBDD s;
(iii) NOT-nodes having a single child;  decomposable (iv) AND-, (v) OR-, (vi) \XOR-, and (vii) \EQUIV-nodes 
$u$ with two children $u_1, u_2$ such that the sub-DAGS $\D_{u_1}$ and $\D_{u_2}$ do not
share any common Boolean variables. Recall that
$\mathbf{\mathbf{\Phi_u}}$ denotes the (vector) subformula of the sub-DAG of $\D$ at a node $u \in \D$; $\D$ finally computes 
the formula $\mathbf{\Phi} = \mathbf{\Phi_r}$ where $r$ is the root of $\D$.
%The same construction holds for all $\Phi \in \mathbf{\Phi}$; below we will discuss the correctness for
 %any arbitrary $\Phi \in \mathbf{\Phi}$; combining all such $\mathbf{\Phi}$s we get the correctness for the ve
%$\mathbf{\Phi}$.

\vspace{0.2in}
\textbf{Step 1: \decDNNFn\ $\D$ for $\mathbf{\Phi}$ to \decDNNFn\ $\OO$ for $\neg \mathbf{\Phi}$.~~} 
As a first step, we convert the  \decDNNFn\ $\D$ into another DAG $\OO$
such that $\OO$ has the same DAG structure as $\D$ but the types of the nodes are changed as follows:

\begin{enumerate}
	\item A 1-sink in $\D$ becomes a 0-sink in $\OO$ and vice versa.
	\item An AND-node in $\D$ becomes an OR-node in $\OO$ and vice versa.
	\item An \XOR-node in $\D$ becomes an \EQUIV-node in $\OO$ and vice versa.
	\item Decision nodes and NOT-nodes remain unchanged.
\end{enumerate}
For any node $u \in \D$, we will denote the corresponding node in $\OO$ by $u'$, 
the subformula at $u$ by $\mathbf{\Phi_u}$ (as before), 
and the subformula at $u' \in \OO$ by $\mathbf{\Psi_u}$.
The following proposition shows that $\OO$ computes $\neg \mathbf{\Phi}$.

					%decision-nodes, NOT-nodes, 
					%\noop\ nodes (nodes with single child that do not perform any operation), 
					%and 
					%decomposable OR-nodes (instead of decomposable AND-nodes as in $\D$).
					%Let $\Psi_u$ denote the subformula at a node $u \in \OO$.
					%For a decomposable OR-node $u$ with two children $u_1, u_2$, $\mathbf{\Psi_u} = \mathbf{\Psi_{u_1}} \vee \mathbf{\Psi_{u_2}}$;
					%further the sub-DAGs $\OO_{u_1}$ and $\OO_{u_2}$  (and therefore $\mathbf{\Psi_{u_1}}$ and $\mathbf{\Psi_{u_2}}$)
					%do not share any common variables. 
			%\par
			%$\D$ and $\OO$ have the same DAG structure (same set of nodes and edges between them). However,
			%(i) for a 1-sink (resp. 0-sink) $u \in \D$, $u$ is a 0-sink (resp. 1-sink) in $\OO$,
			%(ii) if $u$ is a decomposable-AND node in $\D$, $u$ is a decomposable-OR node in $\OO$, 
			%(iii) if $u$ is a NOT-node in $\D$, $u$ is a %\noop\ node 
			%NOT-node in $\OO$, and
			%(iv) if $u$ is a decision-node testing variable $x$ in $\D$ with children $u_0$ and $u_1$ for 0-branch and 1-branch respectively, 
			%then $u$ is also a decision-node testing variable $x$  in $\OO$ with children $u_0$ and $u_1$ for 0-branch and 1-branch respectively.

\begin{proposition}
For all nodes $u \in \D$ and the corresponding $u' \in \OO$, 
$\mathbf{\Psi_u} = \neg \mathbf{\Phi_u}$. 
\end{proposition}
\begin{proof}
 We prove the claim by an induction on 
a \emph{reverse topological order} of the nodes $u$ in $\D$  (a node is processed only after all its descendants are processed)
that $\mathbf{\Psi_u} = \neg \mathbf{\Phi_u}$. 
The hypothesis holds for the 0-sinks and 1-sinks by construction. Suppose the hypothesis holds up to some step
and consider the next node $u$ in the reverse topological order.
\begin{enumerate} [(a)]
\item If $u'$ is a NOT-node in $\OO$ ($u$ is a NOT-node in $\D$), $\mathbf{\Psi_u} = \neg \mathbf{\Psi_{u_1}}$ = $\neg (\neg \mathbf{\Phi_{u_1}})$ (by induction hypothesis) = 
$\neg (\neg (\neg \mathbf{\Phi_{u}})) = \neg \mathbf{\Phi_u}$.

\item If $u'$ is a decomposable OR-node in $\OO$ ($u$ is a decomposable AND-node in $\D$) with children $u_1, u_2$,
then $\Psi_{u} = \mathbf{\Psi_{u_1}} \vee \mathbf{\Psi_{u_2}}$ = $(\neg \mathbf{\Phi_{u_1}}) \vee (\neg \mathbf{\Phi_{u_2}})$ (by induction hypothesis)
= $\neg (\mathbf{\Phi_{u_1}} \wedge \mathbf{\Phi_{u_2}})$ = $\neg \mathbf{\Phi_{u}}$.

\item If $u'$ is a decomposable AND-node in $\OO$ and $u$ is a decomposable OR-node in $\D$ is similar to the above case.

\item If $u'$ is a decomposable \XOR-node in $\OO$ ($u$ is a decomposable \EQUIV-node in $\D$) with children $u_1, u_2$,
then $\mathbf{\Psi_{u}} = \mathbf{\Psi_{u_1}} \cdot \neg \mathbf{\Psi_{u_2}}$ $\vee$ $\neg \mathbf{\Psi_{u_1}} \cdot \mathbf{\Psi_{u_2}}$ = 
$(\neg \mathbf{\Phi_{u_1}}) \cdot \neg (\neg \mathbf{\Phi_{u_2}})$ $\vee$ $\neg (\neg \mathbf{\Phi_{u_1}}) \cdot (\neg \mathbf{\Phi_{u_2}})$ (by induction hypothesis) = 
$(\neg \mathbf{\Phi_{u_1}}) \cdot \mathbf{\Phi_{u_2}}$ $\vee$ $\mathbf{\Phi_{u_1}} \cdot \neg \mathbf{\Phi_{u_2}}$ = 
$\neg (\mathbf{\Phi_{u_1}} \cdot \mathbf{\Phi_{u_2}}$ $\vee$ $\neg \mathbf{\Phi_{u_1}} \cdot \neg \mathbf{\Phi_{u_2}})$
$= \neg \mathbf{\Phi_u}$.

\item If $u'$ is a decomposable \EQUIV-node in $\OO$ and $u$ is a decomposable \XOR-node in $\D$ is similar to the above case.

\item If $u'$ is a decision-node in $\OO$ ($u$ is also a decision-node in $\D$) with two children $u_0, u_1$ on 0- and 1-branch respectively, 
then $\mathbf{\Psi_u} = (\neg x) \mathbf{\Psi_{u_0}} \vee x \mathbf{\Psi_{u_1}}$
= $(\neg x) (\neg \mathbf{\Phi_{u_0}}) \vee x (\neg \mathbf{\Phi_{u_1}})$ (by induction hypothesis)
= $(\neg x) (\neg \mathbf{\Phi}[x = 0]) \vee x (\neg \mathbf{\Phi}[x = 1])$ 
$= \neg \mathbf{\Phi_u}$.
\end{enumerate}
\hfill
\end{proof}

\vspace{0.2in}
\textbf{Step 2: Combining \decDNNFn s $\D$ and $\OO$ into a \decDNNFn\ $\PP$ that does not have any NOT-node~~} 
$\PP$ has the union of nodes of $\D$ and $\OO$. 
If $u \in \D$ is not a NOT-node (\ie, $u' \in \OO$ is also not a NOT-node),
%a decision-node, a 0- or 1-sink (also in $\OO$), or a decomposable-AND node (decomposable-OR node in $\OO$),
then the connections to the child or children of $u$ and $u'$ are retained in $\PP$. 
Let $u \in \D$ be a NOT-node with child $u_1$,
and $u' \in \OO$ the corresponding NODE-node with child $u_1'$.
% (a) find the corresponding $u \in \OO$ and its child $u_1 \in \OO$ (say $u'$ and $u_1'$ respectively), 
Then 
(a) add edge from $u$ to $u_1'$, and make $u$ a \noop\ node that computes the same functions as its unique child $\mathbf{\Phi_u} = \mathbf{\Psi_{u_1}}$,
and (b) add edge from $u'$ to $u_1$, and make $u'$ a \noop\ node.
%Do the same connection for NOT-nodes $u \in \OO$ (add connection to the unique child of $u \in \D$). 
After this construction, $\PP$
will not have any NOT-node, any \noop\ node can be removed by simply connecting all its incoming edges directly to its unique child.
If $r$ is the root of $\D$, the corresponding $r \in \PP$ (and the DAG reachable from $r$ in $\PP$) computes the same function as $\D$.
This conversion increases the number of nodes by a factor of 2.
For any node $u \in \PP$, let $\mathbf{\Theta_u}$ be the subformula at $u$.

\begin{proposition}
%For all nodes $u \in \D$, there is a node $u \in \PP$ computing the formula $\Theta_u$ in $\PP$ at $u$ is 
$\mathbf{\Theta_u}= \mathbf{\Phi_u}$ if $u \in \D$, and $\mathbf{\Theta_u} = \mathbf{\Psi_u} = \neg \mathbf{\Phi_u}$ for the corresponding $u' \in \OO$.
\end{proposition}
\begin{proof}
We prove this by induction on the nodes in $\D$ in reverse topological order. 
The base case holds for the 0- and 1-sinks. The hypothesis holds by construction for any decision-node, or decomposable AND-, OR-\XOR-, or \EQUIV-node $u\in \PP$.
Let $u \in \D$ be a NOT node which is a \noop\ node in $\PP$.  Then $\Theta_u = \mathbf{\Psi_{u_1}} = \neg \mathbf{\Phi_{u_1}} = \mathbf{\Phi_{u}}$.
Similarly $\mathbf{\Theta_{u}} = \mathbf{\Psi_{u}}$ if $u \in \OO$ is a NOT-node.\\
\end{proof}

\vspace{0.2in}
\textbf{Step 3: DAG $\PP$ with decomposable-AND and -OR nodes to \FBDD\ $\F$ with quasipolynomial blow-up.~~}
 This conversion is similar to the one in \cite{Beame+13-uai} where we showed that a \decDNNF\ can be converted to 
an equivalent \FBDD\ with at most quasipolynomial increase in size. The DAG $\PP$ can have additional decomposable-OR nodes
in addition to decision-nodes and decomposable-AND nodes in a \decDNNF. The same construction works here except the 
correctness proof that both the DAGs $\PP$ and $\F$ compute the same formula. Evaluating the function on a given assignment in $\PP$
is different than a \decDNNF. We need to check that 
\begin{itemize}
\item both branches of a decomposable AND-node evaluate to 1; 
\item at least one branch of a decomposable OR-node
evaluates to 1; 
\item exactly one branch of a decomposable \XOR-node evaluates to 1; and 
\item either both or none of the branches of a decomposable \EQUIV-node
evaluates to 1.
\end{itemize}
Here we give a simple construction and proof of correctness.
\par
For any decomposable (AND-, OR-, \XOR-, or \EQUIV-) node   $u \in \PP$ with two children $u_1$ and $u_2$, 
wlog. let the sub-DAG $\PP_{u_1}$ has fewer number of nodes (lighter sub-DAG)
than $\PP_{u_2}$ (heavier sub-DAG). 
For any such node $u$, recursively (top-down) create a \emph{private} copy of the lighter sub-DAG $L_{u} = \PP_{u_1}$
such that any path from the root of $\PP$ to any node $v \in L_u$ goes through $u$. 
It is not hard to see that this replication leads to at most a quasi-polynomial increase in size (see \cite{Beame+13-uai} for a detailed analysis).
Let this intermediate DAG be $\PP'$.
\par
Then we reconnect some of the edges in $\PP'$ to replace the decomposable-AND and -OR nodes by \noop\ nodes, so that they can be removed altogether later.
We process the nodes in reverse topological order. When we process $u$, we maintain the invariant that the sub-DAG at $u$
is an FBDD, \ie, does not have any decomposable nodes but they still compute the same function as in $\PP$ or $\PP'$.
The sinks and their parents (which must be decision-nodes) are unchanged and they ensure that the base case holds.
\par
Suppose we are at a decomposable node $u$, its sub-DAGs $\PP'_{u_1}$ and $\PP'_{u_2}$ have been already changed to \FBDD s $\F_{u_1}$ and $\F_{u_2}$ respectively.
Further, $\F'_{u_1}$ is a private copy for $u$.
\begin{itemize}
\item If $u$ is an AND-node (in $\D$ or $\OO$), %its sub-DAGs $\PP'_{u_1}$ and $\PP'_{u_2}$ have been already changed to \FBDD s $\F_{u_1}$ and $\F_{u_2}$ respectively.
 Then we add all 1-sinks of $\F_{u_1}$ to the root of $\F_{u_2}$ (note that the root of $F_{u_2}$ is a copy of the original node
$u_2 \in \PP$) to simulate the AND-function. 
\item If $u$ is an OR node (in $\D$ or $\OO$), we connect the 0-sinks in of $\F_{u_1}$ to the root of $\F_{u_2}$
to simulate the OR-function.
\item Suppose $u$ is an \XOR-node in $\D$ with light child $u_1$ and heavy child $u_2$. Then there is an \EQUIV-node $u' \in \OO$   
with light  child $u_1'$ and heavy child $u_2'$. Therefore, the sub-DAGs, $\PP'_{u_1}$ and $\PP'_{u_1'}$ are private sub-DAGs of $u$ and $u'$.
\par
To simulate the \XOR-function $\mathbf{\Phi_u}$ = $\mathbf{\Phi_{u_1}} \cdot \neg \mathbf{\Phi_{u_2}}$  $\vee$ $(\neg \mathbf{\Phi_{u_1}}) \cdot \mathbf{\Phi_{u_2}}$,
we (i) connect all the 1-sinks of $\F'_{u_1}$ (private) to  $\F'_{u_2'}$ (shared child of $u'$), (ii) delete the edge $(u', u_2')$,  
(iii) connect all the 0-sinks of $\PP'_{u_1'}$ (private) to  $\PP'_{u_2}$ (shared), and (iv) delete the edge $(u, u_2)$.
Since $\F'_{u_2'}$ computes $\mathbf{\Psi_{u_2}} = \neg \mathbf{\Phi_{u_2}}$ and 
$\F'_{u_1'}$ computes $\mathbf{\Psi_{u_1}} = \neg \mathbf{\Phi_{u_1}}$, 
these connections correctly simulate \XOR-function at $u$.
\par
Similarly, to simulate the \EQUIV-function at $u'$, $\mathbf{\Psi_u}$ = $\mathbf{\Psi_{u_1}} \cdot \mathbf{\Psi_{u_2}}$  $\vee$ $(\neg \mathbf{\Psi_{u_1}}) \cdot (\neg \mathbf{\Psi_{u_2}})$,
we (i) connect all the 1-sinks of $\F'_{u_1'}$ (private) to  $\F'_{u_2'}$ (shared), and
(ii) connect all the 0-sinks of $\F'_{u_1'}$ (private) to  $\F'_{u_2}$ (shared child of $u$).

\item  The case when $u$ is an \EQUIV-node and $u'$ is an \XOR-node is similar to the above case.
\end{itemize}

\begin{proposition}
The sub-DAG $\F_u$ for every node $u$ (and in particular the full DAG at the root) 
created by the above construction is an \FBDD.
\end{proposition}
\begin{proof}
Since the sub-DAGs at $u_1$ and $u_2$ (in $\D, \OO, \PP$ and $\PP'$) did not share any variable, every variable is still tested at most once
along any path in the resulting DAG $\F_u$.
To see that the DAG structure is retained and no cycles are created, 
orient the edges from $u$ to lighter and heavier sub-DAGs such that they are the left and right 
child respectively. 
Also place $u' \in \OO$ and its children $u_1, u_1'$ in $\F$ such that
both $u, u'$ appear at left of both $u_2, u_2'$.
Then the original connections in $\PP'$ are top-down, and the new connections are from left to right, which cannot create a cycle.
\end{proof}
Therefore, after the conversion, we get an \FBDD\ $\F$ with at most quasi-polynomial increase in size and computing the same formula as the original \decDNNFn\ $\D$. %OLD "appendix.tex"
\section{Discussion}
\label{sec:conclusions}
In this paper we proved exponential separations between 
lifted model counting using extensional query evaluation
and state-of-the-art propositional methods for exact model counting.
Our results were obtained by proving exponential lower bounds on the 
sizes of the \decDNNF\ representations implied by those proposition methods
even for queries that can be evaluated in polynomial time.
We also introduced \decDNNFn s, which generalize \decDNNF s while
retaining their good algorithmic properties for model counting.
Though our query lower bounds apply equally to their \decDNNFn\ representations,
\decDNNFn s may prove to be better than \decDNNF s in other scenarios.

%In this paper we 
%showed several lower bounds related to propositional methods for weighted model counting in query evaluation.
%We proved exponential lower bounds on the sizes of \FBDD s for a class of queries that includes
%both easy and \#P-hard queries, which imply exponential lower bounds on the sizes of \decDNNF s and \decDNNFn s
%for these queries,  which in turn imply exponential lower bounds on the running times of all the 
%state-of-the-art \decDNNF-based (DPLL-style) propositional
%model counting algorithms. For easy queries in this class weighted model counting can be done in time polynomial in the size of the database
%using extensional query evaluation; therefore the lower bounds in this paper proves a formal separation between 
%grounded inference and lifted inference techniques. 

%There are several future directions. 
%We introduced \decDNNFn s which generalize \decDNNF s but still
%amenable to a linear-time model counting in the size of the representation. While \decDNNF s are used by all existing
%model counting algorithms implcitly or explicitly to the best of our knowledge, it would be interesting 
%to explore whether the decomposable \XOR- and \EQUIV-nodes can be used by these model counting algorithms
%for some applications.

%On the other hand, our lower bound results
%raises an interesting question:
%
%\begin{itemize}
%\item For which 
  %queries can one perform weighted model counting in polynomial time on
  %the propositional grounding? \label{item:oq1}
%\end{itemize}
%

In light of our lower bounds, it would be interesting to prove a dichotomy,
classifying queries into
those for which any \decDNNF-based model counting algorithm takes
exponential time and those for which such algorithms run in polynomial time. 
In this paper we showed such a dichotomy for a very restricted class of queries.
A dichotomy for general model counting is known for the broader query class UCQ
\cite{DBLP:journals/jacm/DalviS12} which classifies queries as either
\#P-hard or solvable in polynomial time.
Our separation results show that this same dichotomy does not extend to
\decDNNF-based algorithms; is there some other general dichotomy that can be
shown for this class of algorithms?
%as a tiny step towards answering the question; %(\ref{item:oq1}) 
%Does a characterization for efficient propositional methods exist for general queries? 

%Also all the lower bounds in this paper hold for 
%existing \decDNNF-based model counting algorithms. What can be inferred about about the model counting problem
%on propositional grounding  in general?
%NOT LIKELY SOLVABLE UNLESS WE CAN SOLVE MAJOR COMPLEXITY QUESTIONS

% The following two commands are all you need in the
% initial runs of your .tex file to
% produce the bibliography for the citations in your paper.

\cut{
{\footnotesize
\bibliographystyle{abbrv}
\bibliography{bib}  
}
}
% You must have a proper ".bib" file
%  and remember to run:
% latex bibtex latex latex
% to resolve all references

%\subsection{References}

% ****************** APPENDIX **************************************

% Example of an appendix; typically would start on a new page
%pagebreak

%\newpage
\begin{appendix}
%\appendix
%\input{appendix}
\section*{Appendix}
\section*{Omitted Proofs from Section~4}
\label{sec:app-h0h1}
%\section*{Appendix}
%{\textbf{\Huge{Appendix}}}\\

Here we present proofs for Lemma \ref{lem:theta1} in the case where $k$ is even, including when $k = 0$. For any partial path $P$ through an FBDD for $\lin_k$, let $Row(P), Col(P)$, and $\mathcal{P}$ be the same entities as before. The techniques required for the case when $k = 0$ are different, so we will present it separately from the case when $k$ is even and nonzero. In either case, it will suffice to come up with a definition of admissible paths in $\mathcal{P}$ so that 
\begin{enumerate}[(1)]
\item Lemma \ref{lem:uniquepaths} holds for all admissible paths, and 
\item for all $i, j$, the first decision node $u$ at which a admissible path tests a variable of the form $R(i), S_\ell (i, j)$ or $T(j)$ is unforced, using terminology introduced in the proof of Lemma \ref{lem:theta1}. 
\end{enumerate}
Then it is readily checked that using the same proof structure as given in \autoref{sec:h0h1} will yield a complete proof.

We will first prove hardness for $\lin_0$. Our definition of a admissible path is as follows:

\begin{definition}
A path $P \in \mathcal{P}$ in a FBDD computing $\lin_0$ is \emph{forbidden} if there exist $i, j$ so that $R(i) = S(i, j) = 0$ or $S(i, j) = T(j) = 0$ along the path. It is {\em admissible} otherwise.
\end{definition}

That this definition follows property (2) is easily verified. We now demonstrate property (1). This will complete the proof for $k  = 0$.

\begin{lemma}
Lemma \ref{lem:uniquepaths} holds when $k = 0$ with this definition of admissible paths.
\end{lemma}

\begin{proof}	
Suppose $P_1, P_2$ are two distinct admissible paths in $\mathcal{P}$ with $\lin_0 [P_1] = \lin_0 [P_2] = F$. Let $u$ be the first node (it must be a decision node) where $P_1$ and $P_2$ diverge, then by Lemma \ref{lem:RTdetermined} it must be of the form $S_1(i, j)$. WLOG assume $P_1$ sets $S(i, j) = 0$ and $P_2$ sets $S(i, j) = 1$. Then the 3-prime implicant $R(i) S(i, j) T(j)$ cannot appear in $F$, nor can $R(i) T(j)$ and there are no units at all, so in particular $P_2$ must set either $R(i) = 0$ or $T(j) = 0$ which implies that $P_1$ is forbidden.
\end{proof}

We will now prove hardness for $\lin_k$ for $k = 2m$ even and nonzero. Here our definition of a admissible path is as follows:

\begin{definition}
A partial path $P \in \mathcal{P}$ in a FBDD computing $\lin_k$ for $k$ even and nonzero is \emph{admissible} if for all $i, j$ it is consistent with one of the four following assignments:
\begin{enumerate}
\item
$R(i) = 1$, $T(j) = 0$ and $S_\ell (i, j) = 0$ for all $\ell$ odd and $S_\ell (i, j) = 1$ for all $\ell$ even
\item
$R(i) = 0$, $T(j) = 1$ and $S_\ell (i, j) = 1$ for all $\ell$ odd and $S_\ell (i, j) = 0$ for all $\ell$ even
\item
$R(i) = T(j) = 0$ and $S_\ell (i, j) = 1$ for all $\ell$ odd and $S_\ell (i, j) = 0$ for all $l$ even
\item
$R(i) = T(j) = 1$ and $S_1(i, j) = S_\ell (i, j) = 0$ for $l$ even and $S_\ell (i, j) = 1$ for $\ell > 1$ odd. %%%%%JERRY, PLZ CHECKeven.
\end{enumerate}
$P$ is \emph{forbidden} if it is not admissible.
\end{definition}

Again that this definition follows property (2) is easily verified, and it suffices to demonstrate property (1).

\begin{lemma}
Lemma \ref{lem:uniquepaths} holds when $k > 0$ and even with this definition of admissible paths.
\end{lemma}

\begin{proof}
Suppose $P_1, P_2$ are two distinct admissible paths in $\mathcal{P}$ with $\lin_k [P_1] = \lin_k [P_2] = F$. Let $u$ be the first (decision) node where $P_1$ and $P_2$ diverge. Then by Lemma \ref{lem:RTdetermined} it must be of the form $S_\ell (i, j)$ for some $l$. WLOG assume $P_1$ sets $S_\ell (i, j) = 0$ and $P_2$ sets $S_\ell (i, j) = 1$. First consider the case when $\ell = 1$. Then since the FBDD follows the unit clause rule we conclude that $P_2$ and hence $P_1$ sets $R(i) = 0$ which implies $P_1$ is forbidden, which is a contradiction.  Hence $l \neq 1$.

Now let $l > 1$. Then the prime implicant $S_{\ell - 1}(i, j) S_\ell (i, j)$ is missing from $\lin_k [P_1] = F$; as no $1$-prime implicants can exist, this implies that $S_{\ell - 1} (i, j) = 0$ in $P_2$. As long as $\ell - 1 > 1$, by the definition of admissible paths we have $P_1$ cannot set $S_{\ell - 1} (i, j) = 0$ so reversing roles and inducting downwards we conclude that either $P_1$ sets $S_2 (i, j) = 0$ and $P_2$ cannot set $S_2(i, j)$ to $0$, or vice versa. WLOG assume that the former case occurs. This implies that $P_2$ must set $S_1 (i, j) = 0$, and by the same logic as above, this implies $P_1$ cannot set $R(i)$ to 0 and 
sets $S_1 (i, j) = S_2(i, j) = 0$. Therefore $P_1$ can potentially be an admissible path only by case 4 above.
%%%%%JERRY, PLEASE CHECK FROM HERE
Then going upward from $S_\ell(i, j)$ to $T(j)$, assuming $\ell > 1$ is odd (and $k$ is even), $S_k(i, j) = 0$ on $P_2$ and $S_{k-1}(i, j) = 0$ on $P_1$.
Therefore, $P_1$ cannot set $S_k(i, j)$ to 0 and must set $T(j)$ to 0, which violates the condition in case 4. 
Now consider an even $\ell > 1$.
If $k = 2$, the only scenario is $\ell= k = 2$. 
Then $S_k(i, j) = 0$ on $P_1$, and $S_k(i, j) = 1$ on $P_2$ by assumption, therefore $T(j)$ must be 0 on $P_2$
and also on $P_1$, which violates case 4. 
If $k > 2$ and $\ell$ is even,  $S_k(i, j) = 0$ on $P_1$, and $S_{k-1}(i, j) = 0$ on $P_2$, and since $k > 2$,  $S_k(i, j)$ cannot be set to 0 on $P_2$, so $T(j) = 0$ on $P_2$, and therefore on $P_1$, 
which violates case 4. 
\end{proof}

\end{appendix}

\end{sloppypar}
\end{document}